\documentclass[
hf,
]{ceurart}
\usepackage[utf8]{inputenc}
\usepackage{lmodern}  % or use 'cm-super' if you prefer
\usepackage{url,hyperref}
\usepackage{array}
\usepackage[notext]{stix2}
\usepackage{amsmath,amstext,stmaryrd}
\usepackage{mathtools}

\usepackage[ruled,vlined,linesnumbered]{algorithm2e}

\usepackage{xcolor}
\usepackage{latexcolors}
\usepackage{graphicx}

\usepackage{latexsym}
\usepackage{subcaption}
\usepackage{caption}
\usepackage{tikz}

\usetikzlibrary{graphs}
\usetikzlibrary{positioning}
\usetikzlibrary{arrows.meta,automata,positioning}
\usetikzlibrary{decorations.pathmorphing}
\usetikzlibrary{petri}
\usetikzlibrary{shapes}
\usetikzlibrary{calc}
\usetikzlibrary{shadows}
\usetikzlibrary{patterns, patterns.meta}

\usepackage{datetime2}
\usepackage{scrtime}

\tikzstyle{place}=[circle,draw,minimum size=2mm]
\tikzstyle{transition}=[rectangle,draw,fill=black,minimum width=6mm,inner ysep=1pt]
\tikzset{
	every token/.style={minimum size=1pt},
	token distance=3pt
	anchor=base,
}
\tikzset{
    vertex/.style={
      draw, 
      circle,
      fill,
      inner sep=0pt,
      outer sep=0pt,
      minimum size=0pt,
    },
    pixel/.style={
      circle,
      draw,
      top color=blue,
      bottom color=blue,
      minimum size=0.4cm 
    },
    transpixel/.style={
      circle,
      draw,
      top color=white,
      bottom color=blue!50,
      minimum size=0.4cm 
    },
    none/.style={
      minimum size=0.4cm 
    }
}
\tikzset{every picture/.style={scale=0.75}}

\countdef\sizeofgrid=4  % the 4 is mandatory but the count is not 4! WHY?
\pgfmathsetcount\sizeofgrid{4}
\dimendef\vrcholsize=0
\dimendef\arrowwidth=0
\newcommand\gridsize{4}

\tikzstyle{vertex}=[
                     top color=white,
                     minimum size=\vrcholsize,
                   ]
\tikzstyle{dagnode}=[ vertex, ]
\tikzstyle{picnode}=[
                     dagnode,
                     bottom color=blue!40,
                     draw,
	            ]
\tikzstyle{picnodetrans}=[
                     dagnode,
                     bottom color=blue!10,
	            ]
\tikzstyle{picborder}=[
                     dagnode,
                     bottom color=blue!10,
      	             node font=\tiny,
                      ]
\tikzstyle{pixel}=[
                   picnode,
                   top color=blue,
                   draw,
                  ]
\tikzstyle{picarrow}=[
  		     line width=\arrowwidth,
		    ]
\tikzstyle{picedget}=[ picedge, \tcolor,   ]
\tikzstyle{picedgeT}=[ picedge, \Tcolor,   ]
\tikzstyle{bfaedge}=[  picedge, \bfacolor, ]
\tikzstyle{rfaedge}=[  picedge, \rfacolor, ]

\newcommand\pictouch{
  \pgfmathsetlength\vrcholsize{+7.06cm}
  \pgfmathsetlength\arrowwidth{+.65cm}
  \tikzstyle{picedge}=[
                       picarrow,
                       -{[length=1pt]},
  		      ]
  \tikzstyle{dagnode}=[ vertex, ]
}
\newcommand\picedge{
  \pgfmathsetlength\vrcholsize{5.95cm}
  \pgfmathsetlength\arrowwidth{.6pt}
  \tikzstyle{picedge}=[
                       line width=.65cm,
                       picarrow,
                       -{Latex[length=4pt,width=13pt]},
		       line width=6pt,
  		      ]
\tikzstyle{dagnode}=[
                     minimum size=4.5mm,
                     top color=white,
                     bottom color=blue!40,
	            ]
}
\newcommand\picarrow{
  \pgfmathsetlength\vrcholsize{+.000001cm}
  \pgfmathsetlength\arrowwidth{+.6pt}
  \tikzstyle{dagnode}=[ vertex, ]
  \tikzstyle{picedge}=[ -latex, ]
  \tikzstyle{dagnode}=[ vertex, circle, ]
}

\newcommand\sigmaX{\sigma}
\newcommand\rhoDA{\rho_A}
\newcommand\rhoOTA{\rho_M}

\newcommand\dontcareone{\qdown}
\newcommand\dontcaretwo{\qright}
\newcommand\alphaone{\alpha_1}
\newcommand\alphatwo{\alpha_2}
\newcommand\alphaonetwo{\alphaone\alphatwo}
\newcommand\betaone{\beta_1}
\newcommand\betatwo{\beta_2}
\newcommand\betaonetwo{\betaone\betatwo}
   \newcommand\statealphaone{\dontcareone\alphaone}
   \newcommand\statealphatwo{\alphatwo\dontcaretwo}
\newcommand\qdown{q_{\outdown}}
\newcommand\qright{q_{\outright}}

\newcommand\forallzeroij{
  \forall i\in[m]_0, j\in[n]_0:
}

\newcommand\OTAsimDAforalltext{
\OTAsimDA \text{ for all } \qdown, \qright \in \dagalphlambda
}
\newcommand\OTAsimDA{
\OTAsimDAinR\implies\OTAsimDAdelta
}

\newcommand\OTAsimDAinR{
(\alphaonetwo \sigmaarrowX \betaonetwo) \in R
}
\newcommand\OTAsimDAdelta{
    \betaonetwo\in
    \delta\left( \statealphaone, \statealphatwo, \sigmaX \right)
}

\newcommand\cmdunused[2]{
  $$
  \forall q\inright,q\indown \in Q :
  ( \qzdag   q\indown \sharparrow \qfdag),
  (q\inright  \qzdag  \sharparrow \qfdag),
  ( \qfdag   q\indown \sharparrow \qfdag),
  (q\inright  \qfdag  \sharparrow \qfdag) \in R
  $$
}

\newcommand\tbd[1]{#1}

\let\mod\relax\DeclareMathOperator{\mod}{{}\mathrm{mod}{}}

\newcommand\dagDef{(V, E, \lab, \IN, \OUT)}
\newcommand\dagdef{G = (V, E, \lab, \IN, \OUT)}
\newcommand\dagdefi{G = (V, E \cup E_\driven, \lab, \IN, \OUT)}

\newcommand{\anchor}[2]{\tikz[remember picture,baseline=-.3ex,inner xsep=0,inner ysep=2pt]{\node (#1) {$#2$};}}
\usetikzlibrary{calc}
\newcommand*\circled[1]{%
  \tikz[%
        node distance=7pt,
        baseline=(vertexlabel.base),
        edgestring/.style={inner sep=0pt},
        vertex/.style={
                       draw,
                       shape=rounded rectangle,
                       inner ysep=0pt,
                       inner xsep=-2pt
                      }
       ]{
    \node[vertex] (vertexlabel) {\rule[-3pt]{0pt}{\dimexpr2ex+2pt}#1};
    \node[edgestring, left=of vertexlabel]  (ingoing)  {};
    \node[edgestring, right=of vertexlabel] (outgoing) {};
    \draw [->>] (ingoing) -- (vertexlabel);
    \draw [->>] (vertexlabel) -- (outgoing);
}}
\newcommand{\lab}{\mathit{\ell}}
\newcommand{\IN}{\mathit{in}}
\newcommand{\OUT}{\mathit{out}}
\newcommand{\src}{\mathit{src}}
\newcommand{\tar}{\mathit{tar}}

\newcommand{\rarrow}[1]{\circled{$\mathtt{#1}$}} % r arrow = rule arrow

\newcommand{\sigmaarrow}{\circled{$\sigma$}}
\newcommand{\sigmaarrowX}{\circled{$\sigmaX$}}

\newcommand{\sharparrow}{\circled{\#}}
\newcommand{\Barrow}{\circled{$\stateB$}}
\newcommand{\Warrow}{\circled{$\stateW$}}

\newcommand{\pixelarrow}{\circled{$\bullet$}} % emptyarrow
\newcommand{\transarrow}{\circled{$\circ$}} % emptyarrow

\newcommand\drule{(\alpha\sigmaarrow\beta)}

\newcommand\sharprule{(\alpha\sharparrow\beta)}
\newcommand\deltaqarrows{
  \delta(q\indown, q\inright, \sigma)
}

\newcommand{\emptygraph}{\varnothing}

\newcommand\pcolor{tyrianpurple}

\newcommand\qcolor{pakistangreen}

\newcommand\tcolor{navajowhite}
\newcommand\Tcolor{jonquil}
\newcommand\rfacolor{kellygreen}
\newcommand\bfacolor{mahogany}
\newcommand\colorp{\color{\pcolor}}

\newcommand\colorq{\color{\qcolor}}

\newcommand\p{{\colorp p}}
\renewcommand\P{{\colorp p'}}
\newcommand\q{{\colorq q}}

\newcommand\symbolDiag{\backslash}
\newcommand\Ldia{L_{\symbolDiag}}
\newcommand\Lstripes{L_\Equiv}
\newcommand\Lbalance{L_\mathrm{balance}}
\newcommand\symboloutdown{\mapsdown}
\newcommand\symboloutright{\barrightarrowdiamond}
\newcommand\symbolindown{\downarrowbar}
\newcommand\symbolinright{\rightarrowbar}
\newcommand\outdown{^\symboloutdown}
\newcommand\inright{^\symbolinright}
\newcommand\indown{^\symbolindown}
\newcommand\outright{^\symboloutright}
\newcommand\pos{_{i,j}}    % current position
\newcommand\posL{_{i,j-1}} % position left
\newcommand\posU{_{i-1,j}} % position up

\newcommand\eijd{e\pos\outdown}
\newcommand\eijr{e\pos\outright}
\newcommand\eiju{e\posU\outdown}

\newcommand\eijl{e\posL\outright}

\newcommand\Eall{E_\text{all}}
\newcommand\einputdriven{E_{\driven}}
\newcommand\edownleft{\hat E_{\downarrow\dottedsquare}}
\newcommand\edownright{\hat E_{\dottedsquare\downarrow}}
\newcommand\Erfa{\hat E_{\text{rfa}\swarrow}}
\newcommand\Ebfa{\hat E_{\downarrow\text{bfa}\downarrow}}

\newcommand\ErightrightArrows{{\hat E}_\rightrightarrows}

\newcommand\ErightleftArrows{{\hat E}_\rightleftarrows}

\newenvironment{griddag}[1][1]{
  \begin{tikzpicture}[scale=#1]
    \foreach \i in {0,...,\sizeofgrid}
      \foreach \j in {0,\gridsize}{ %sizeofgrid does compile only for range
        \node [picborder] (x\i y\j) at (\i,\j) {$\#$};
        \node [picborder] (x\j y\i) at (\j,\i) {$\#$};
      }
  \pgfmathaddtocount\sizeofgrid{-1} % hack TODO refactor
  \foreach \x in {1,...,\sizeofgrid}
     \foreach \y in {1,...,\sizeofgrid}
        \node [picnode] (x\x y\y) at (\x,\y) {};
}{
  \end{tikzpicture}
}
\newenvironment{griddagTrans}[1][1]{
  \begin{tikzpicture}[scale=#1]
    \foreach \i in {0,...,\sizeofgrid}
      \foreach \j in {0,\gridsize}{ %sizeofgrid does compile only for range
        \node [picborder] (x\i y\j) at (\i,\j) {$\#$};
        \node [picborder] (x\j y\i) at (\j,\i) {$\#$};
      }
  \pgfmathaddtocount\sizeofgrid{-1} % hack TODO refactor
  \foreach \x in {1,...,\sizeofgrid}
     \foreach \y in {1,...,\sizeofgrid}
        \node [picnodetrans] (x\x y\y) at (\x,\y) {};
}{
  \end{tikzpicture}
}
\newcommand\cooDAGedges{
  \pgfmathaddtocount\sizeofgrid{1} % hack TODO refactor
  \foreach \x [remember=\x as \lastx (initially 0)] in {1,...,\sizeofgrid}
     \foreach \y in {0,...,\sizeofgrid}
       \path (x\lastx y\y) edge [picedge,carmine] node {} (x\x y\y); 
  \foreach \y [remember=\y as \lasty (initially 0)] in {1,...,\sizeofgrid}
     \foreach \x in {0,...,\sizeofgrid}
       \path (x\x y\y) edge [picedge,carmine] node {} (x\x y\lasty);
}

\newcommand\cooDAGedgesTrans{
  \pgfmathaddtocount\sizeofgrid{1} % hack TODO refactor
  \foreach \x [remember=\x as \lastx (initially 0)] in {1,...,\sizeofgrid}
     \foreach \y in {0,...,\sizeofgrid}
       \path (x\lastx y\y) edge [picedget] node {} (x\x y\y); 
  \foreach \y [remember=\y as \lasty (initially 0)] in {1,...,\sizeofgrid}
     \foreach \x in {0,...,\sizeofgrid}
       \path (x\x y\y) edge [picedget] node {} (x\x y\lasty);
}

\newcommand\Erightright{
  \foreach \x [remember=\x as \lastx (initially 0)] in {1,...,\sizeofgrid}
     \foreach \y in {0,...,\sizeofgrid}
       \path (x\lastx y\y) edge [picedge,black] node {} (x\x y\y); 
}

\newcommand\Edownrfa{
  \foreach \y [remember=\y as \lasty (initially 0)] in {1,...,\sizeofgrid}
       \path (x4y\y.south) edge [\rfacolor] node {} (x0y\lasty.north); % TODO \sizofgrid does not compile
}
\newcommand\Elinesrfa{
  \foreach \x [remember=\x as \lastx (initially 0)] in {1,...,\sizeofgrid}
     \foreach \y in {0,...,\sizeofgrid}
       \path (x\lastx y\y.north east) edge [picedge,\rfacolor] node {} (x\x y\y.north west); 
}
\newcommand\Edownleft{
  \foreach \y [remember=\y as \lasty (initially 0)] in {1,...,\sizeofgrid}
       \path (x0y\y) edge [picedge] node {} (x0y\lasty); % TODO \sizofgrid does not compile
}
\newcommand\Edownright{
  \foreach \y [remember=\y as \lasty (initially 0)] in {1,...,\sizeofgrid}
       \path (x4y\y) edge [picedge] node {} (x4y\lasty); % TODO \sizofgrid does not compile
}

\newcommand\Esese{
  \foreach \x [remember=\x as \lastx (initially 0)] in {1,...,\sizeofgrid}
    \foreach \y [remember=\y as \lasty (initially 0)] in {1,...,\sizeofgrid} 
       \path (x\lastx y\y) edge [picedge,\rfacolor] node {} (x\x y\lasty); 
}

\newcommand\Eleftbottom{
  \foreach \x [remember=\x as \lastx (initially 0)] in {1,...,\sizeofgrid}
       \path (x\x y0) edge [picedge] node {} (x\lastx y0); 
}

\newcommand\grassDAGedges{
  \pgfmathaddtocount\sizeofgrid{1} % hack TODO refactor
  \Erightright
  \Edownleft
}

\newcommand\Erightleftarrowssouth{
  \foreach \x [remember=\x as \lastx (initially 0)] in {1,...,\sizeofgrid}
  {
     \foreach \y in {1,3,...,\sizeofgrid}
       \path (x\lastx y\y.south east) edge [bfaedge] node {} (x\x y\y.south west); 
     \foreach \y in {0,2,...,\sizeofgrid}
       \path (x\x y\y.south west) edge [bfaedge] node {} (x\lastx y\y.south east); 
  }
}

\newcommand\Eoxdown{
  \foreach \y [remember=\y as \lasty (initially 0)] in {1,...,\sizeofgrid} {
     \pgfmathsetmacro\x{int(4 * (mod(\y, 2)))};
       \path (x\x y\y) edge [picedge,\bfacolor] node [left] {} (x\x y\lasty);
  }
}
\newcommand\pixelize[1]{
  \foreach \x in {1,...,\sizeofgrid}
    \foreach \y in {1,...,\sizeofgrid}
        \node [pixel] at (#1) {};
}

\newcommand{\picdags}{\mathcal P}
\newcommand{\hatpicdags}{\hat{\mathcal P}}
\newcommand{\allpicdags}{\picdags\cup\hatpicdags}
\newcommand\Pij{P\pos}

\newcommand\Pdia{P_\backslash}
\newcommand\hPdia{\hat P_\backslash}

\newcommand{\symbolqzdag}{\pushout} % top + left line, dot; z = zero
\newcommand{\symbolqfdag}{\pullback} % bottom+right line, dot; f= final
\newcommand{\symbolqzfdag}{\boxdot} % square+dot for start+final state
\newcommand{\qzdag}{\symbolqzdag} % top + left line, dot; z = zero
\newcommand{\qfdag}{\symbolqfdag} % bottom + right line, dot
\newcommand{\qzfdag}{\symbolqzfdag} % square+dot for start+final state

\newcommand{\driven}{\iota}

\DeclareMathOperator{\aDAG}{\dottedsquare} % some DAG encoding
\DeclareMathOperator{\eqdag}{\equiv_\dottedsquare} % equiv some DAG enc
\DeclareMathOperator{\linDAG}{\squarehfill} % coordinate DAG
\newcommand{\symbolrfa}{\rightarrowshortleftarrow} %\rightthreearrows % 
\newcommand{\symbolbfa}{~^{\hookleftarrow}_{\hookrightarrow}} % 
\DeclareMathOperator{\rfaDAG}{\symbolrfa} % coordinate DAG
\DeclareMathOperator{\eqrfa}{\equiv_{\symbolrfa}} % equiv coordinate DAG
\newcommand{\symbollinl}{\downarrow\rightrightarrows} % 
\DeclareMathOperator{\linlDAG}{\symbollinl} % 
\DeclareMathOperator{\linlDAGL}{\mathcal L_{\symbollinl}} % 
\DeclareMathOperator{\ddalinl}{\DDA_{\symbollinl}} % 
\DeclareMathOperator{\ddalinr}{\DDA_{\symbollinr}} % 
\newcommand{\symbollinr}{\rightrightarrows\downarrow} % 
\DeclareMathOperator{\linrDAG}{\symbollinr} % 
\DeclareMathOperator{\linrDAGL}{\mathcal L_{\symbollinr}} % 
\DeclareMathOperator{\bfaDAG}{\symbolbfa} % coordinate DAG
\DeclareMathOperator{\eqbfa}{\equiv_{\symbolbfa}} % equiv coordinate DAG
\DeclareMathOperator{\cooDAG}{\squarehvfill} % coordinate DAG
\DeclareMathOperator{\eqcoo}{\equiv_\squarehvfill} % coordinate DAG
 
\newcommand{\symboldia}{\squarenwsefill} %\rightthreearrows % 
\newcommand{\symbolDia}{\squareneswfill} %\rightthreearrows % 
\DeclareMathOperator{\diaDAG}{\symboldia} % same diagonlly
\newcommand\stateWB{\circlebottomhalfblack}
\newcommand\stateBW{\circletophalfblack}
\newcommand\stateB{\mdlgblkcircle}
\newcommand\stateW{\mdlgwhtcircle}

\newcommand\sesearrow{\searrow\!\!\!\searrow}

\newcommand\Sigmastar{\Sigma^*} 
\newcommand\Sigmastarstar{\Sigma^{\ast,\ast}}
\newcommand\Sigmasharp{\Sigma \cup \{\#\}}
\newcommand\dagalph{N}

\newcommand\dagalphlambda{\dagalph \cup \{\lambda\}}

\newcommand\automatonfont[1]{\mathsf{#1}}

\newcommand\DDA{\automatonfont{DDA}}% (top-down)deterministic DAG automaton
 
\newcommand\NDA{\automatonfont{NDA}} % nondeterministic DAG automaton
\newcommand\FSA{\automatonfont{FSA}}
\newcommand\DFA{\automatonfont{DFA}}
\newcommand\NFA{\automatonfont{NFA}}
\newcommand\RFA{\automatonfont{RFA}} 
 
\newcommand\BFA{\automatonfont{BFA}}

\newcommand\twoDOTA{\automatonfont{2DOTA}} 
\newcommand\twoOTA{\automatonfont{2OTA}}

\usepackage{subcaption} % subfigure in eptcs and ceurart
\usepackage{amsthm}     % load myself for eptcs
\usepackage{thmtools}
\declaretheorem[numberwithin=section]{theorem}

\newtheorem{corollary}[theorem]{Corollary}

\newtheorem{definition}[theorem]{Definition}

\newtheorem*{pdconjecture*}{Pushdown Conjecture}

\newtheorem{example}[theorem]{Example}

\title{%
A Unifying Approach to Picture Automata
}

\usepackage{listings}
\begin{document}

\copyrightyear{2025}
\copyrightclause{Copyright for this paper by its authors.
  Use permitted under Creative Commons License Attribution 4.0
  International (CC BY 4.0).}
\conference{ITAT'25: Information Technologies -- Applications and Theory, September 26--30, 2025, Telg{\'a}rt, Slovakia}

\author[1]{%\DTMnow{}
Yvo Ad~Meeres}[%
email=Yvo.AdMeeres@mailbox.org,
]
\cormark[1]
\address[1]{%
University of Bremen,
Department of Theoretical Computer Science, Bibliothekstr. 5, 283 59  Bremen, Germany}

\author[2]{Franti{\v s}ek Mr\'az}[%
email=frantisek.mraz@mff.cuni.cz,
]
\cormark[1]
\address[2]{Charles University, Department of Software and Teacher Training, Malostransk\'e n\'am.~25,
118 00 Prague 1, Czech Republic}

\cortext[1]{Corresponding author.}

\begin{keywords}
directed acyclic graph \sep
DAG automaton \sep
picture language
\sep
regular language
%\sep
%two-dimensional automaton \sep
%finite state automaton \sep
%automata theory
\end{keywords}

\newcommand{\titlerunning}{\DTMnow}

\maketitle

\begin{abstract}
%A directed acyclic graph (DAG)
%can encode a two-dimensional string, called a picture.
%We propose a framework for recognizing picture languages via directed acyclic graph automata (DAG automata) by encoding two-dimensional inputs into directed acyclic graphs. There are many ways how to encode a string or a picture into a DAG. We consider input-agnostic encodings, where the underlying DAG depends just on the dimensions of the input, and input-driven encodings, where the in- and out-degree of a node depends on the symbol at the corresponding position in the input string or picture. Three distinct input-agnostic encodings characterize classes of picture languages accepted by returning finite automata (row-wise traversal), boustrophedon automata (bidirectional processing), and online tessellation automata (diagonal tiling). When encoding a string as a simple path from its first to its last symbol, DAG automata recognize just regular languages. On the other hand, input-driven encodings enable recognition of some context-sensitive string languages. In two dimensions, input-driven encodings enable DAG automata to surpass online tessellation automata.
A directed acyclic graph (DAG) can represent a two-dimensional string or picture. We propose recognizing picture languages using DAG automata by encoding 2D inputs into DAGs. An encoding can be input-agnostic (based on input size only) or input-driven (depending on symbols). Three distinct input-agnostic encodings characterize classes of picture languages accepted by returning finite automata, boustrophedon automata, and online tessellation automata. Encoding a string as a simple directed path limits recognition to regular languages. However, input-driven encodings allow DAG automata to recognize some context-sensitive string languages and outperform online tessellation automata in two dimensions.
\end{abstract}

%\tableofcontents\newpage

% stix2 package test: Pitchfork $\pitchfork$

%einput{fig/enc}

%\marginpar{TODO
%- extending locality
%- relax neighborhood but still finite
%- representing pics as graphs for unifying approach
%- learning picture lang \cite{learnpicautomata}
%- add Anis' paper about graphs for image recognition
%}
%
%\cite{daggrammar}
%\cite{DBLP:conf/mol/Drewes17}

A picture can be represented as a rectangular array of symbols. Thus, pictures can be considered as the two-dimensional extension of strings. A picture language is then a set of pictures. While many different automata models that accept two-dimensional inputs were introduced (see \cite{DBLP:reference/hfl/GiammarresiR97}), there is still no class of such automata whose corresponding class of picture languages has properties similar to the class of regular languages.

%A picture is a rectangular array of symbols, extending strings to two dimensions, and a picture language is a set of such pictures. Although various automata models for two-dimensional inputs exist (see \cite{DBLP:reference/hfl/GiammarresiR97,DBLP:journals/jcss/FernauPST18}), none define a class of picture languages with properties comparable to regular languages.

Probably the closest class of picture languages with respect to its properties is the class of recognizable languages (REC), which are languages accepted by two-dimensional online tessellation automata ($\twoOTA$). Actually, each language from REC can be obtained as a projection of a local picture language, where a picture language is local if it is the set of all pictures having all subpictures of dimension two-by-two from a given finite set.

%The class of recognizable languages (REC), accepted by two-dimensional online tessellation automata ($\twoOTA$), is likely the closest to having regular-like properties. Each REC language can be seen as a projection of a local picture language, defined as all pictures whose two-by-two subpictures belong to a given finite set.

DAG automata introduced by
Kamimura and Slutzki in \cite{kamimura-slutzki:81} process single-rooted DAGs,
just like the DAGs we propose for encoding strings and pictures.
%just as will be the DAGs suggested
Discussions at a Dagstuhl Seminar on graph transformations
with insights from Bj{\"o}rklund and Maletti \cite{DagRep.5.3.143}
gave rise to a variant of a DAG automaton subsequently introduced by
Blum and Drewes~\cite{journals/iandc/BlumDrewes2019}.
%This model exhibits favorable properties.
A DAG automaton has a finite set of states, a finite alphabet, and a set of rules. For an input DAG, the DAG automaton assigns states to the edges of the input DAG and checks whether all nodes comply with the rules of the DAG automaton. In the positive case, the automaton accepts the input graph. Otherwise, it rejects. So, the DAG automaton, similarly to a $\twoOTA$, checks only local properties (labeling of the incoming and outgoing edges for each node).

%DAG automata, introduced by Kamimura and Slutzki in \cite{kamimura-slutzki:81}, process single-rooted DAGs like those used for encoding strings and pictures. Inspired by discussions at a Dagstuhl Seminar and work by Blum and Drewes~\cite{journals/iandc/BlumDrewes2019}, a DAG automaton has a finite set of states, an alphabet, and rules. It assigns states to edges and checks if all nodes follow its rules, accepting or rejecting the input graph. Like $\twoOTA$, DAG automata verify only local properties based on edge labeling at each node.

Regarding one-dimensional input, it is easy to represent a word (a string of symbols) as a graph. We can, e.g., take the set of positions in the string as the set of nodes labeled with the corresponding symbols of the string and connect them in a path from the first to the last. In two dimensions, we can represent a picture as a rectangular grid of nodes connected using horizontal and vertical edges. If we use oriented edges, a directed acyclic graph (DAG) representing a given picture, as illustrated in Fig.~\ref{fig:enc}, can be built easily. We call such encodings input-agnostic, as the edges of the encoding DAG do not depend on the contents of the input.

%A graph can represent a word using its positions as nodes labeled with symbols, connected in a path from the first to the last. Similarly, a picture can be represented as a rectangular grid of nodes linked by directed horizontal and vertical edges. Such a DAG encoding is called input-agnostic, since its edges depend only on the input shape, not the content.

Simple input-agnostic DAG encodings for strings yield automata recognizing only regular languages, but for pictures, different input-agnostic encodings yield DAG automata corresponding to known 2D automata like returning finite automata, boustrophedon automata, and online tessellation automata.

We also introduce input-driven encodings of strings and pictures into DAGs, where a node's in-degree and out-degree depend on its label. Thus, DAG edges can vary between different pictures even if they share the same dimensions.

%Input-driven encodings increase the power of DAG
%automata when recognizing both string and picture languages. Input-driven encodings enable recognition of
%some context-sensitive string languages. In two dimensions, input-driven DAG automata accept a proper
%super-class of the class of recognizable languages.

Input-driven encodings enhance DAG automata's power, enabling recognition of some context-sensitive string languages and a proper super-class of REC in two dimensions.

DAG automata have many appealing properties. Blum and Drewes \cite{journals/iandc/BlumDrewes2019} proved that emptiness and
finiteness are decidable for DAG automata, and the class of DAGs accepted by DAG automata is closed
under union and intersection. These properties carry over to the case of accepting encodings of picture
languages.
Deterministic DAG automata can be reduced and minimized, and their equivalence is decidable in polynomial time \cite{journals/iandc/BlumDrewes2019}. Equivalence for deterministic DAG automata with respect to picture language recognition is still an
open problem, as two deterministic DAG automata accepting the same encodings of pictures need not
be equivalent, because they can accept different sets of DAGs that do not encode pictures.
Ad Meeres \cite{AdMeeres2024} provides a framework to port known efficient finite state automata algorithms and
properties from string to DAG languages. By applying these techniques, we could obtain classes of picture
languages that can be efficiently recognized by DAG automata.

DAG automata have many strong properties~\cite{journals/iandc/BlumDrewes2019}: emptiness and finiteness are decidable, and accepted DAG classes are closed under union and intersection, including for picture language encodings. We can minimize deterministic DAG automata, and their equivalence is polynomial-time decidable, though equivalence for deterministic DAG automata with respect to picture language recognition is still open. A recent framework by Ad Meeres \cite{AdMeeres2024} enables adapting efficient finite automata algorithms from strings to DAG languages.

The paper is structured as follows. After establishing the basic notation,
Sect.~\ref{sec:2Dautomata} provides an overview of common automata accepting two-dimensional inputs. Sect.~\ref{subsec:dagautomata} presents a model of a DAG automaton suitable for recognizing pictures due to its limited capabilities. Sect.~\ref{sec:pictureencoding} introduces various encodings of pictures into DAGs. Sect.~\ref{sec:comp} compares DAG automata first to classical string automata to show the capabilities of DAG automata in the 1D case and then to the picture automata for the 2D case. The concluding Sect.~\ref{sec:conclusions} summarizes the obtained results and indicates open problems for future research. The majority of the proofs have been placed in the Appendix. %%% FINAL: change wording
\begin{figure*}[t]
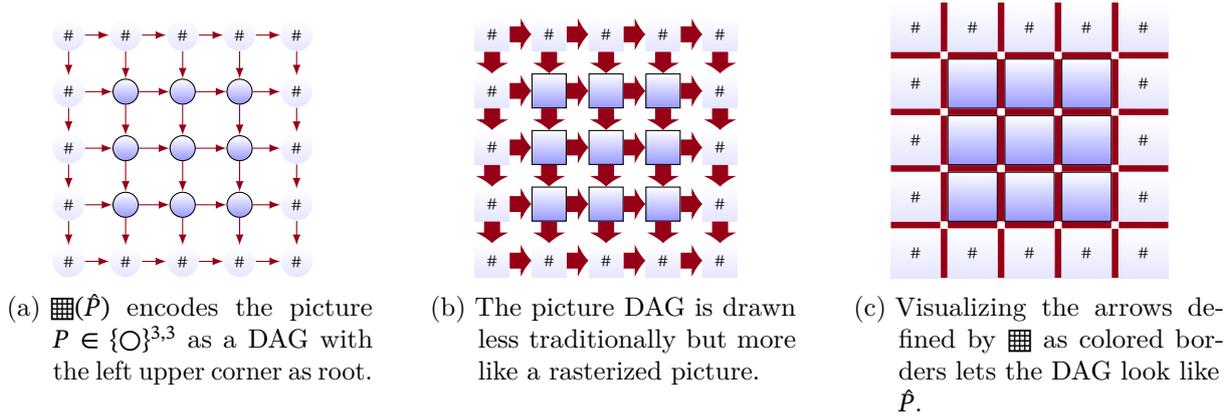

\captionsetup{subrefformat=parens}
\begin{subfigure}{0.3\textwidth}\center
\input{pic/enc-arrow}
\caption{
    $\cooDAG(\hat P)$ encodes the picture $P\in\{\stateW\}^{3,3}$
    as a DAG with the left upper corner as root.\\
    }\label{pic:enc-arrow}
\end{subfigure}
\hfill
\begin{subfigure}{0.3\textwidth}\center
\input{pic/enc-edge}
\caption{The picture DAG is
    drawn less
   %like DAGs are illustrated
    traditionally
    but more like a rasterized picture.\\
  }\label{pic:enc-edge}
\end{subfigure}
\hfill
\begin{subfigure}{0.3\textwidth}\center
\input{pic/enc-touch}
\caption{Visualizing the arrows defined by $\cooDAG$
    as colored borders lets the DAG
    look like %the mere
    %almost like the border picture
    $\hat P$.
  }\label{pic:enc-touch}
\end{subfigure}
\caption{
The DAG encodes a picture $P$ of dimensions $(3,3)$
over the unary alphabet $\Sigma = \{\stateW\}$
as a \emph{picture DAG} with the edges specified
by the \emph{DAG encoding} $\cooDAG$.
A DAG encodes a picture
illustrated with vertices in the form of circles as shown in
%\subref{pic:enc-arrow}.
\subref*{pic:enc-arrow}.
A transition between the traditional illustration of graphs and rasterized pictures is the picture DAG in
\subref*{pic:enc-edge}.
The visualization of the DAG,
omitting the orientation of the arrows is shown in
\subref*{pic:enc-touch}.
}
\label{fig:enc}

\end{figure*}

% For testing subfigures
%\begin{figure}
%\centering
%\begin{subfigure}{0.3\textwidth}
%${CAT_WITH_A}$
%\caption{ This is a cat} \label{cat2}
%\end{subfigure}
%\begin{subfigure}{0.3\textwidth}
%${ELEPHANT_WITH_B}$
%\caption{This is an elephant} \label{elephant2}
%\end{subfigure}
%\captionsetup{subrefformat=parens}
%\caption{Two animals: \subref{cat2} a cat,
%and \subref{elephant2} an elephant}
%\label{animals}
%\end{figure}
%

The paper is structured as follows. After establishing the basic notation,
Sect.~\ref{sec:2Dautomata} provides an overview of common automata accepting two-dimensional inputs. Sect.~\ref{subsec:dagautomata} presents a model of a DAG automaton suitable for recognizing pictures due to its limited capabilities. Sect.~\ref{sec:pictureencoding} introduces various encodings of pictures into DAGs. Sect.~\ref{sec:comp} compares DAG automata first to classical string automata to show the capabilities of DAG automata in the 1D case and then to the picture automata for the 2D case. The concluding Sect.~\ref{sec:conclusions} summarizes the obtained results and indicates open problems for future research.

The majority of the proofs have been placed in the
Appendix~\ref{sec:proofappendix}.
\section{Notation}
\label{sec:notation}
For the positive integers $\mathbb{N}$
with $\mathbb{N}_0 = \mathbb{N}\cup\{0\}$,
%we define
let
$[n]   = \{1,\dots,n\}$ and
$[n]_0 = \{0,\dots,n\}$
where $n\!\in\!\mathbb{N}$.
The cardinality of a set $S$ equals  $|S|$.
A finite set $S$ of \emph{symbols} forms an \emph{alphabet}.
Concatenated symbols yield a \emph{string}.
The concatenation of two strings $u$ and $v$ is written
as $u\cdot v$ or as their juxtaposition~$uv$.
For a string $w=w_1w_2 \dots w_i \dots w_n \in S^*$ with $i \in [n]$,
symbol $w_i \in S$ occurs at \emph{position} $i$ where
$|w|$ denotes the length $n$ of it,
$\lambda$ represents the empty string,
$|w|_\sigma$ gives the number of occurrences of the symbol $\sigma$ in $w$
and $[w]$ denotes the smallest subset $T\subseteq S$ such that $w \in T^*$.
All finite strings \emph{over} $S$ form the set $S^*$,
with any subset of $S^*$ called a \emph{(string) language},
and $S^n \subset S^*$ comprising all strings over $S$ of length $n$.
A doubly ranked alphabet \(\Sigma\) is defined as an alphabet
in which each symbol is assigned a rank
through a function \(r: \Sigma \to \mathbb{N}_0^2\).
%An alphabet $\Sigma$ is doubly ranked if each symbol is assigned a rank
%through a function $r: \Sigma \to \mathbb{N}_0^2$, denoted as $(\Sigma,r)$.
% $[w]$ denotes the smallest set $S$ such that $w\in S^*$.

%TODO: A for set and A for automaton
An \emph{automaton model}~$M$ is a computational model
%TODO: gen only if we use it in the paper
that defines how a specific instance of $M$, an automaton $A$,
\emph{recognizes} a language $L(A)$.
%Two automata $A$ and $A'$ are equivalent denoted by $A\equiv A'$
%  if they recognize or generate the same language, thus $L(A) = L(A')$.
Two automata models $M$ and $M'$ are considered equivalent,
denoted by $M\equiv M'$,
if they recognize the same class of languages.
Throughout the paper,
%an automaton $A$ is defined by a tuple using the symbols
in an automaton definition, the following symbols
do not require further definition when they appear:
$\Sigma$ for an input alphabet excluding the border signal %$\#$,
$\# \notin \Sigma$, %for a special symbol marking border of a picture,
$Q$ for a finite set of states,
$q_0 \in Q$ for an initial state,
$F \subseteq Q$ for a set of final states.
A  finite state automaton ($\FSA$), in its nondeterministic variant ($\NFA$)
(or its deterministic one ($\DFA$)) is defined
as a tuple $A = (Q, \Sigma, \delta, q_0, F)$,
where $\delta \subseteq Q \times \Sigma \times Q$
is the transition relation (or a function in the case of DFA).
The reflexive transitive closure
$\delta^* \subseteq Q \times \Sigma^* \times Q$ is defined by:
(i) for all $q \in Q$, $(q,\lambda,q) \in \delta^*$, and
(ii) if $(q,a,q') \in \delta$ and $(q',w,q'') \in \delta^*$
for $a \in \Sigma$, $w \in \Sigma^*$, then $(q,aw,q'') \in \delta^*$.
If $(q_0, w, q') \in \delta^*$ and $q' \in F$,
$A$ \emph{accepts} $w \in \Sigma^*$.
The automaton $A$ recognizes the language $L(A)$
which is the set of all strings that $A$ accepts.

\section{Computational Models for Two Dimensions}
\label{sec:2Dautomata}

%write sophisticated text that 2-ary this binary relations
%can be viewed as two-dimensional.
Both
rectangular arrays of symbols
and other binary, i.e.\,2-ary, relations
can be regarded as two-dimen\-sional structures.
For both notions of two-dimensionality,
automata have been introduced in the literature.
Picture automata process two-dimensional matrices over an alphabet
(see Sect.~\ref{subsec:picautomata}),
while DAG automata digest finite multisets of an alphabet
with a binary relation defined on them
(see Sect.~\ref{subsec:dagautomata}).
% The literature introduces automata to address
%these two notions of two-dimensionality.
%An n-ary relation can be considered n-dimensional.
%An n-ary relation consists of tuples with n components.
%Each component can be thought of as a dimension,
%meaning that the overall relation occupies an n-dimensional space.
%For example, in a binary relation (2-ary), you can visualize it as a two-dimensional space (like a coordinate system).

\subsection{Picture Automata}
\label{subsec:picautomata}
Inspired from natural languages,
a one-dimensional string is called a word.
Inspired from image processing,
a two-dimensional string or a matrix is called a \emph{picture}.
Both terms are standard in automata theory.
A \emph{picture automaton} recognizes a \emph{picture language}.

We consider the three picture automaton models
\emph{returning} and \emph{boustrophedon finite automaton}
($\RFA$ and $\BFA$, Def.~\ref{def:rfabfa}) %,
and
\emph{two-dimensional online tessellation automaton}
($\twoDOTA$ and $\twoOTA$, Def.~\ref{def:ota}). %,
%and the %sgraffito automaton %(2SA, Def.~\ref{def:sgraffito}).
We consolidate these fundamentally different picture automata
under a single DAG-automaton model,
providing a complete characterization for each individual automaton
while simultaneously unifying them within a common framework,
thereby capturing the full expressive power of
each characterized picture automaton.

Notably, when restricted to one-row input pictures
(i.e. in the one-dimensional case),
each of these picture automaton models
recognizes exactly the regular string languages.

\begin{definition}[(Boundary) Picture]\label{def:pic}
A \emph{two-dimensional string} $P$ over a finite alphabet $\Sigma$ is called a \emph{picture}.
The \emph{dimensions} of a picture $P$ with $m$ rows and $n$ columns are given by the ordered pair
$(m,n)$.
%$\left(\rows(P), \cols(P)\right)$.
The symbol $\Lambda$ denotes the unique
\emph{empty picture} with dimensions $(0, 0)$.
$\Sigma^{m,n}$ represents the set of all pictures
with dimensions $(m,n)$ over $\Sigma$.
The notation $P_{i,j}$ refers to the symbol from $\Sigma$
at the \emph{position} $P(i,j)$ of the picture $P$,
where $i \in [m]$ denotes the \emph{row} counted from top to bottom
and   $j \in [n]$      the \emph{column} counted from left to right:
%that is the intersection of the $i$-th row and $j$-th column in $P$,
%and is called a \emph{pixel}.
\begin{displaymath}
P =
\begin{array}{ccc}
P_{1,1} & \ldots & P_{1,n} \\
\vdots & \ddots & \vdots \\
P_{m,1} & \ldots & P_{m,n}
\end{array}.
\end{displaymath}
Any subset of the set of all pictures over $\Sigma$, denoted $\Sigma^{*,*}$, is called a \emph{picture language}.
Given a picture $P\in \Sigma^{m,n}$, its \emph{boundary picture}
%the remaining elements of $\hat{P}$ are $\#$.
$\hat{P}$ is the picture $P$ itself surrounded by the border symbol $\#$,
thus a picture over $\Sigma \cup \{\#\}$ with dimensions $(m+2,n+2)$
where the rows and columns range over $[m+1]_0$ and $[n+1]_0$, respectively.
Let $\hat{P}_{i,j} = P_{i,j}$,
for all $i \in [m]$ and $j \in [n]$,
and
$\hat{P}_{i,j} = \#$,
% for all $\{i,j\} \cap \{ 0, m+1, n+1\} \neq \emptyset$.
% this would put # in both rows n+1, m+1
for all $i \in [m+1]_0$ and $j \in [n+1]_0$, where $i \in \{0,n+1\}$ or $j\in\{0,m+1\}$.

\def\vshift{0.5}
\def\gapp{1.5}
\def\vertProdLength{2.5}
\def\gappB{2*\gapp + \vertProdLength}
\begin{center}
\parbox[c]{0.2\textwidth}{
\begin{tikzpicture}
\draw (0,0) -- (4,0) -- (4,3-\vshift) -- (0,3-\vshift) -- (0,0);
\draw (0,0.5) -- (4,0.5);
\draw (0,2.5-\vshift) -- (4,2.5-\vshift);
\draw (0.5,0) -- (0.5,3-\vshift);
\draw (3.5,0) -- (3.5,3-\vshift);
\node at (2,1.5-\vshift/2) {$P$};
\node at (0.25,0.25) {$\#$};
\node at (0.25,2.75-\vshift) {$\#$};
\node at (3.75,0.25) {$\#$};
\node at (3.75,2.75-\vshift) {$\#$};
\node at (0.25,0.75) {$\#$};
\node at (0.25,2.25-\vshift) {$\#$};
\node at (0.25,1.65-\vshift/2) {$\vdots$};
\node at (3.75,0.75) {$\#$};
\node at (3.75,2.25-\vshift) {$\#$};
\node at (3.75,1.6-\vshift/2) {$\vdots$};
\node at (0.75,0.25) {$\#$};
\node at (1.25,0.25) {$\#$};
\node at (2.75,0.25) {$\#$};
\node at (3.25,0.25) {$\#$};
\node at (2.0,0.25) {$\ldots$};
\node at (0.75,2.75-\vshift) {$\#$};
\node at (1.25,2.75-\vshift) {$\#$};
\node at (2.75,2.75-\vshift) {$\#$};
\node at (3.25,2.75-\vshift) {$\#$};
\node at (2.0,2.75-\vshift) {$\ldots$};
\end{tikzpicture}}
$\quad =
\quad
\hat P \quad =
\begin{array}{cccccc}
\hat P_{0,0} && \ldots && \hat P_{0,n+1} \\
\hat P_{1,0} & P_{1,1} & \ldots & P_{1,n} & \hat P_{1,n+1} \\
\vdots && \ddots && \vdots \\
\hat P_{m,0} & P_{m,1} & \ldots & P_{m,n} & \hat P_{m,n+1} \\
\hat P_{m+1,0} && \ldots && \hat P_{m+1,n+1}
\end{array}
$
\end{center}

\end{definition}

Both automata models defined below, $\RFA$ and $\BFA$,
were introduced by Fernau, Paramasivan, Schmid and Thomas~%
\cite{DBLP:journals/jcss/FernauPST18}.
%Both $\RFA$ and $\BFA$ 
They
work as a finite state machine
(as a $\DFA$ in the deterministic and an $\NFA$ in the nondeterministic case)
equipped with a scanning strategy
%; c.f. scanning strategy for two-dimensional restarting tiling automata~%
\cite{DBLP:journals/ijfcs/PrusaM13}
which enables the automaton models to operate on two-dimensional input.
%which can be viewed as a total order $\prec$
%an the elements of the two-dimensional input string.
The scanning strategy serializes a boundary picture $\hat P$ into
a string $\hat p$ over $\Sigmasharp$.
Afterward, the picture $P$ is accepted
if the corresponding $\FSA$ accepts this serialized picture $\hat p$.
Both models scan their input picture horizontally line by line,
according to their respective scanning strategy%.
,but
%However,
they differ in their scanning directions.
An $\RFA$ proceeds uniformly from left to right,
this corresponds to the reading of Western texts
where the eye has to `jump'
from the end of a line to the beginning of the next one.
A $\BFA$, on the other hand, proceeds like an ox plowing a field,
alternating its direction.

\begin{definition}[$\RFA$, $\BFA$]\label{def:rfabfa}
Both a (deterministic) \emph{returning finite automaton}      ($\RFA$)
and a  (deterministic) \emph{boustrophedon  finite automaton} ($\BFA$)
are given by a 6-tuple
$ A = ( Q , \Sigma, \delta , q_0, F , \# )$,
where
%%%% $Q$ is the finite set of states,
%%%%  $\Sigma$ is the input alphabet,
$\delta \subseteq Q \times  (\Sigma \cup \{\#\}) \times  Q$
is the transition relation (transition function
$\delta : Q \times  (\Sigma \cup \{\#\}) \to Q$).

Both automata models are equipped with a \emph{scanning strategy},
which defines a total order on the elements of a two-dimensional string.
Formally, for a two-dimensional string $S$ holds
$S(i_1,j_1) \prec S(i_2,j_2)$ if and only if, in case of an
\begin{itemize}
\item $\RFA:$
  $(i_1+j_1 < i_2+j_2) \lor (j_1 = j_2 \land i_1 < i_2)$
  and in case of a
%$(i < k) \lor (i = k \land j < l)$
\item $\BFA:$
$
  (i_1+j_1 < i_2+j_2) \lor
\bigl((i_1+j_1 \mod 2 = 0 \land j_1 > j_2) \lor
     (i_1+j_2 \mod 2 = 1 \land j_1 < j_2)\bigr)
$.
\end{itemize}
With
$M = ( Q , \Sigmasharp, \delta , q_0, F )$ being an $\NFA$ ($\DFA$),
the language $L$ accepted by $A$ is given as
\begin{displaymath}
\begin{array}{r@{\hskip 2pt}l@{\hskip 2pt}l}
L(A) = \{ & P \in \Sigma^{m,n}
\mid
%&
\text{ for all }
i \in [m+1]_0 \text{ and } j \in [n+1]_0
\\
&
s_1 s_2
\text{\dots}
s_k
s_{k+1}
\text{\dots}
s_{(m+2)(n+2)}
\in L(M)
\text{ with }
\hat P_{i_1j_1} \prec \hat P_{i_2j_2}
\text{ where }
s_k     = \hat P_{i_1j_1}
\text{ and }
s_{k+1} = \hat P_{i_2j_2}
&\}\\
\end{array}
\end{displaymath}
using the respective total order $\prec$ as specified above.
\end{definition}

We presented an alternative, yet equivalent, definition
of \cite{DBLP:journals/jcss/FernauPST18},
which serves as a characterization emphasizing both
the treatment of the boundary symbols as well as the scanning strategy.
It is easy to see that our modified definition is correct since it
does not change the expressive power of the automata models:
every state computed within the top or bottom boundary
could be computed within the first, resp. last row of the picture as well
and two consecutive boundary symbols $\#\#$ do not equip
an automaton with expressiveness since it can do the same in one step
while reading the squeezed single $\#$, see below.

Here, we defined scanning strategies on a boundary picture
with all four sides explicitly represented.
The original definition, in contrast,
delimits pictures only on the left and right and even then,
it does not reference all those boundary symbols.
% In \cite{DBLP:journals/jcss/FernauPST18}, both the returning automaton and boustrophedon automaton start from the top left corner of the input picture and ends after scanning the whole last row of the picture.  In our modified definition the boustrophedon automaton starts at the top-right corner of the boundary picture.
Using boundary pictures with the full boundary, on all four sides,
uniformly across all automaton models supports a coherent perspective
and facilitates their systematic comparison in Sect.~\ref{sec:comp}.
By limiting the boundary symbols to acting as a `fence' in the serialized
string, we mimic the behavior of the original definition easily:
\begin{displaymath}
{
\setlength{\arraycolsep}{2pt} % increase space between columns
\begin{array}{cllc}
L(A) = \{ & P \in \Sigma^{m,n}
\mid
%&
\text{ for all }
i \in [m+1] \text{ and } j \in [n+1]
\\
&
s_1 s_2
\text{\dots}
s_k
s_{k+1}
\text{\dots}
s_{m\cdot(n+1)-1}
\in L(M)
\text{ with }
\hat P_{i_1j_1} \prec \hat P_{i_2j_2}
\text{ where }
s_k     = \hat P_{i_1j_1}
\text{ and }
s_{k+1} = \hat P_{i_2j_2}
&\}\\
\end{array}
}
\end{displaymath}
Interpreting the strict order $\prec$ above as a partial order $\preceq$
by regarding consecutive positions labeled with $\#$ as equal,
like $P(i,n+1) = P(i+1,0)$, allows us to take canonical representatives
of the resulting equivalence classes. We take the boundary symbols
of the right border for both models as canonical representatives.
In this way, in above language-defining formula,
$j$ no longer ranges over zero, excluding the left boundary.
E.g., in case of an $\RFA$,
for every pair $P(i,n+1)$ and $P(i+1,0)$,
the canonical representative is $P(i,n+1)$ where $i\in[m+1]_0$.
While the corresponding formalism may seem intricate at first glance,
the underlying idea is entirely natural and aligns closely
with the structure imposed by the scanning strategy,
shown in Sect.~\ref{subsubsec:rfabfa}.
This partial order $\preceq$ mimics the original definition of
$\RFA$ and $\BFA$~\cite{DBLP:journals/jcss/FernauPST18}
where all vertices labeled $\#$ are present,
but, only the canonical representative is referenced in a run
which would be $P(i,n+1)$ for an $\RFA$.

By making the scanning strategy explicit, our characterization
reveals its effect on the family of languages recognized by the models.
As a result of the explicit scanning strategy,
Sect.~\ref{subsubsec:dia} demonstrates that parametrizing it
enables the recognition of a wider range of picture languages.

% \begin{theorem}[RFA $\equiv$ $\BFA$ \cite[Theorem 2]{DBLP:journals/jcss/FernauPST18}]
% TODO: cite with option does not work in theorem name :(
\begin{theorem}[$\RFA$ $\equiv$ $\BFA$ \cite{DBLP:journals/jcss/FernauPST18}]
\label{theorem:rfaeqbfa}
$\RFA$ and $\BFA$ recognize the same picture language family.
\end{theorem}
 % #citeproof
\begin{theorem}[det. $\RFA$ $\equiv$ nondet. $\RFA$, det. $\BFA$ $\equiv$ nondet. $\BFA$ {\cite[Lemma 14]{DBLP:journals/jcss/FernauPST18}}]\label{theorem:bfa-deteqnondet}
$\RFA$ and $\BFA$ have the same expressive power in their deterministic
and in their nondeterministic variant.
\end{theorem}

 % #citeproof
%\input{theorem/1dRFAeqL3}
%\marginpar{$(F\smallsetminus Q)$ and set of initial states $I$ -- check source of this copy paste def}
The computation of an online tessellation automaton on a picture $P \in \Sigmastarstar$ consists of associating a state $q_{i,j}
\in Q$ to each position $\hat P(i,j)$
depending only on the two states at the positions above and left from $\hat P(i,j)$.
%Such state is given by the transition relation $\delta$ and depends on the states already associated to positions $P(i-1,j)$ and $P(i,j-1)$ and on the symbol $P_{i,j}$.
%
\begin{definition}[$\twoOTA$, $\twoDOTA$
                   \cite{DBLP:reference/hfl/GiammarresiR97}]
\label{def:ota}
A nondeterministic (deterministic)
\emph{two-dimensional online tessellation automaton} $A$
referred to as $\twoOTA$ ($\twoDOTA$), is defined by
the five-tuple $A = (\Sigma, Q, q_0, F, \delta)$
where %:
%$\Sigma$ is an input alphabet,
%$Q$ is a finite set of states,
%$q_0 \in Q$ is an initial state,
%$F \subset Q$ is a set of final states,
$\delta:Q\times Q\times \Sigma \to 2^{Q}$
($\delta:Q\times Q\times \Sigma \to Q$) is a transition relation (function).
A run of $A$ on an input picture $P \in \Sigma^{*,*}$
is a function assigning a state $q\in Q$
to every position of its boundary picture $\hat P(i,j)$.
The state of a run $\rho$ at position $\hat P(i,j)$
is referenced as $\rho(\hat P(i,j)) = q_{i,j}$.
A run $\rho$ of $A$ on a picture $P$
assigns the initial state $q_0$ to all positions
in the first column and the first row of its boundary picture $\hat{P}$.
That is, $q_{0,j} = q_{i,0} = q_0$,
for all $i \in [m+1]_0$ and $j \in [n + 1]_0$.
If $A$ is deterministic,
$\rho$ assigns a state $q_{i,j}$ to each position $P(i,j)$,
for $i \in [m]$ and $j \in [n]$ with
$q_{i,j} = \delta(q_{i-1,j}, q_{i,j-1}, P_{i,j})$;
if $A$ is nondeterministic, $\rho$ guesses all states such that
$q_{i,j} \in  \delta(q_{i-1,j}, q_{i,j-1}, P_{i,j})$.
It accepts $P$, if there exists a run $\rho$ such that $q_{m,n} \in F$.
\end{definition}

Here, we are interested in DAG automata. Blum and Drewes \cite{journals/iandc/BlumDrewes2019} showed that DAG automata have many desirable properties. The class of languages accepted by DAG automata is closed under union and intersection, but not under complementation. Emptiness and finiteness for languages accepted by DAG automata are decidable in polynomial time. Deterministic DAG automata can be minimized and tested for equivalence in polynomial time.

A deterministic DAG automaton ($\DDA$) for strings encoded as strings DAGs corresponds exactly to an  $\NFA$ that can only guess that it is at the end of the string. A DAG automaton is aware of the termination of the input earlier than a string automaton because it can infer that a symbol is the last one when it detects no outgoing edges, unlike a string automaton, which only realizes the word has ended after the last symbol has already been consumed. Thus, apart from this nondeterministic guessing of the end of a string, a $\DDA$ acts on a string DAG like a $\DFA$ on a string. However, we will add start and end markers to simplify the proof. Then, a one-dimensional $\DDA$ works precisely the same way as a $\DFA$.
\begin{theorem}[\cite{DBLP:journals/jcss/FernauPST18,DBLP:reference/hfl/GiammarresiR97}]
\label{theorem:1dpicAeqFSA}
The picture automata $\RFA$, $\BFA$, $\twoDOTA$, and $\twoOTA$
run on one-dimensional pictures $P\in \Sigma^{m,1}$, thus on strings,
recognize exactly the regular string languages.
\end{theorem}
 % #citeproof

\subsection{Graph Automata}
\label{subsec:dagautomata}
A graph is a set with a 2-ary relation on it.
Consequently, we introduce DAG-automata for parsing two-dimensional strings.
We limit our focus to DAGs due to the inherent orientation of binary tuples, and the acyclic restriction offers algorithmic benefits.
We choose the most limited model of a DAG automaton defined in the literature
\cite{journals/iandc/BlumDrewes2019}
to ensure that the model is not overly powerful.
It operates on multigraphs with ordered edges.
%This model operates on multigraphs with ordered edges.

\begin{definition}[DAG]\label{def:dag}
A \emph{directed acyclic graph over $\Sigma$}, abbreviated as \emph{DAG},
is a tuple $\dagDef$
with $\,\Sigma$, $V\!$, and $E$ being disjoint finite sets,
\emph{alphabet of vertex labels}, the sets of \emph{vertices} and \emph{edges},
respectively.
The vertices are labeled by $\lab \colon V \rightarrow \Sigma$.
An \emph{edge} $e \in E$ joins two distinct vertices $v, w \in V\!$.
The \emph{source} $v$ is referenced by $\src(e)$ and
the \emph{target} $w$ by $\tar(e)$\footnote{Note that $E$ can contain multiple edges connecting the same pair of nodes. Such distinct edges with same source and same target can even cross since the sources and targets are ordered.}.
By $\IN, \OUT \colon V \rightarrow E^*$, we assign to each vertex $v \in V$
its \emph{incoming} and \emph{outgoing} edges
such that
$\src(e) = v \Leftrightarrow e \in [\OUT(v)]$
and
$\tar(e) = v \Leftrightarrow e \in [\IN(v)]$.
%Frantishek:
%Here, $[s]$, for a string $s$, denotes the set of elements of $s$.
%outcommented by yvo since already in notations
These edges are ordered as specified
by the strings $\IN(v)$ and $\OUT(v)$.
For the graphical representation of a graph,
each vertex is drawn as a circle.
The ingoing edges are placed along the upper half of the circle and
the outgoing edges along the lower half, arranged
from left to right according to the strings $\IN(v)$ and $\OUT(v)$.%
\footnote{Hence, the ingoing edges appear clockwise and
the outgoing edges counterclockwise around the circle.}
For every sequence $e_1 e_2 \cdots e_n$
of edges $e_1,\dots,e_n\in E$ in a DAG
with $n\in\mathbb N$, it holds that $v_0 \neq v_n$ in
$\{\src(e_i),$ $\tar(e_i)\}= \{v_{i-1}, v_i\}$ for all $i\in[n]$,
meaning that the DAG is \emph{acyclic}.
The \emph{empty graph} $\emptygraph$ is the graph with $V=\emptyset$.
A vertex is called a \emph{root}
if $\IN(v) = \lambda$
and a \emph{leaf}
if $\OUT(v) = \lambda$.
% moved and finegrained to degenc, no cannot move since needed before sec 4
A DAG $S = (\{v_0,\dots,v_n\}, \{e_1,\dots,e_n\}, \lab, \IN, \OUT)$
\emph{encodes} a string $\sigma_0 \sigma_1 \dots \sigma_n$
as a \emph{string DAG} if
$\ell(v_i) = \sigma_i$ for $i\in[n]_0$,
$\IN(v_i) = e_{i}$ for $i\in[n]$ and
$\OUT(v_{i-1}) = e_i$ for $i\in[n]$.
%$|\IN(v)| \leq 1$ and $|\OUT(v)| \leq 1$ for all vertices $v \in V$.
\end{definition}

\begin{definition}[DAG Automaton]
\label{def:dagautomaton}%
%TODO: add empty graph to language of DAG automaton
%in order to be in line / consistent with regular string languages

A \emph{DAG automaton} is a triple $A = (Q, \Sigma, R)$,
% outcommented alphabet and states since defined in notations ∀ automata:
% where $Q$ is a finite set of states, $\Sigma$
% is a finite alphabet and
where
$R$ is a finite set of \emph{rules}.
A state $q\in Q$ may also be called an \emph{edge label}.
A rule $r \in R$ is either of the form $\drule$,
where $\sigma \in \Sigma,$ or of the form $\emptygraph$\footnote{%
Adding the empty graph to the rule set allows consistency
with the string languages
but differs from the definition of DAG automata in~%
\cite{journals/iandc/BlumDrewes2019}.}
and the \emph{head} $\alpha$ and the \emph{tail} $\beta$ are elements of $Q^*$.
% the following addition is wrong and must be omitted
%    such that in $\alpha$ and $\beta$ any state can occur at most once.
% TODO yvo check this
A \emph{run} of $A$ on a DAG $G = (V, E, \lab, \IN, \OUT)$
is a mapping $\rho\colon E \rightarrow Q$,
extended to strings $\rho\colon E^* \rightarrow Q^*$ by
applying $\rho$ component-wise to every edge $e \in E$,
such that, for every vertex $v \in V$, $\left(\rho(\IN(v))\rarrow{\lab(\mbox{$v$})}\rho(\OUT(v))\right) \in R$.
%\footnotetext{Recall that we extend functions such as $\rho$ in the canonical way to sequences and sets.}
$A$ \emph{accepts} a nonempty DAG $G$ if such a run exists, and $A$ accepts the empty DAG $\emptygraph$ if $\emptygraph$ is in $R$.
The \emph{DAG language recognized by $A$} is
$L(A)=\{\,G\mid A\text{ accepts } G \text{ and } G \text{ is connected}\}$.

The DAG automaton $A$ is called
\emph{top-down deterministic} if
%\emph{top-down (bottom-up) deterministic} if
for every fixed combination of $\sigma \in \Sigma$, $\alpha \in Q^*$,
and $b \in \mathbb{N}_0$ there exists at most one
$\beta\in Q^b$ such that $\drule \in R$.\footnote{Intuitively,
this means that for a vertex
with fixed vertex label, fixed degree, and known ingoing edge labels
only one labeling of the outgoing edges is allowed by $R$.
In such a case, an input DAG permits at most one run, and this run can be constructed deterministically from the roots downward (the leaves upward).
}
A \emph{rule cycle}\footnote{%
See \cite{AdMeeres2024} for a more formal definition.}
is a ring of rules
(a sequence where the first rule is adjacent to the last)
where in each pair of adjacent rules a state $q\in Q$ is connected
to the same state $q$, once in a head and once in a tail.
%Stronger, it is called \emph{vertex deterministic} if,
%for every vertex label $\sigma\in\Sigma$
%and integers $a,b \in \mathbb{N}_0$, there exist
%at most one pair $(\alpha,\beta)$, such that $\alpha \in Q^a$,
%$\beta \in Q^b$ and $\drule \in R$.\footnote{Note,
%that this means top-down and bottom-up determinism.}
%
%The DAG language accepted by $A$ is then said to
%be top-down / bottom-up / vertex deterministic.
%be top-down / bottom-up / vertex deterministic.
We abbreviate a top-down deterministic DAG automaton as $\DDA$
and a nondeterministic one as $\NDA$.

\end{definition}

For a DAG automaton $A$, let $L_{\mathrm{onerow}}(A)$ denote the set of string DAGs accepted by the DAG automaton $A$, and $\mathcal{L}_{\mathrm{onerow}}(DAG) = \{ L_{\mathrm{onerow}}(A) \mid A \mbox { is a DAG automaton} \}$ is the class of string languages that encoded as strings DAGs are accepted by DAG automata.
%\begin{theorem}\label{theorem:1dDDAeqL3}
%\input{theorem/1dDDAeqL3} % #moved2appendix
%\end{theorem}
\begin{restatable}{theorem}{theoremdDDAeqL}
\label{theorem:1dDDAeqL3}%

%\begin{theorem}[for string DAGs: $\DDA\equiv\DFA$]\label{theorem:1dDDAeqL3}
%\end{theorem}
%\begin{proof}
%\newcommand\AW{A_{DFA}}
%\newcommand\AN{A_{NFA}}
%\newcommand\AD{A_{DDA}}
%% TODO compl result:Then without exponential blowup construct the set of $F$ for a $\DFA$???.
%Given a regular language $L_3 \in \mathcal L_3$.
%By assumtion,
%for every regular string language
%we have both a $\DFA$ and a $\DDA$ recognizing it.
%Due to closedness under quotient operation,
%we can add start and end markers to a regular language,
%thus $L=\{\#w\# \mid w \in L_3\}$.
%A $\DFA$ recognizes a regular language thus given a $\DFA$
%$\DFA$ $\AW = \dfadefNFA$ with $L(\AW)=L$,
%we show the one-to-one correspondence to a $\DDA$
%$A_D = \dagadefQ$:
%\begin{itemize}
%%\item The alphabet $\Sigma$ is the same for both automata.
%%$Q_W = Q_D\cup\{\lambda_\alpha\}\cup\{\lambda_\beta\}$ % two \cup instead of comma in case $\lambda_\alpha=\lambda_\beta$
%\item
%$Q_D = Q_W \setminus\{\lambda_\alpha,\lambda_\beta\}$
%\item $(\alpha,\sigma,\beta) \in \delta \iff \drule \in R$
%		for $\alpha,\beta\notin
%		\{\lambda,\lambda_\alpha,\lambda_\beta\}$
%\item $(\lambda_\alpha,\#,\beta) \in \delta \iff (\lambda\sharparrow\beta) \in R$
%\item $(\alpha,\#,\lambda_\beta) \in \delta \iff (\alpha\sharparrow\lambda) \in R$
%%\item $q_0=\lambda_\alpha$
%%\item $F=\{\lambda_\beta\}$.
%\end{itemize}
%
%\end{proof}
%
%\hr
%\begin{theorem}\label{theorem:1dDDAeqL3}
$\mathcal{L}_{\mathrm{onerow}}(DAG) = Reg$, where $Reg$ is the class of regular languages.
%\end{theorem}
 % #moved2appendix
\end{restatable}

\section{Encoding Strings and Pictures as DAGs}
\label{sec:pictureencoding}

How to encode a picture as a graph?  The obvious idea is
to represent all positions in a picture as vertices of a DAG.
We call this a picture DAG.
However, this gives us an arbitrary, even possibly disconnected DAG.

\begin{definition}[Picture DAG]\label{def:picdag}
A \emph{picture DAG} $G = (V, E, \ell, \IN, \OUT)$
encodes a two-dimensional string $S$
with positions $S(i,j)$ and $i,j\in\mathbb N$
by
interpreting the positions within $S$
as the uniquely referenced vertices of $G$ with
$V = \{ S(i,j) \mid i \in [m] \text{ and } j \in [n] \}$,
labeled as $\ell(S(i,j)) = S_{i,j}$.
The set of edges is arbitrary
(i.e., neither $\IN$, $\OUT$ nor $E$, is specified)
unless specified by a picture-to-DAG encoding.
The set of all picture DAGs over an alphabet $\Sigma$
given by the context
is denoted as $\picdags = \{G \mid G \mbox{ encodes a picture from } P\in\Sigmastarstar\}$;
the set of the boundary pictures is
$\hat\picdags = \{G \mid G \mbox{ encodes a picture }  \hat P, \mbox{ where } P\in\Sigmastarstar\}$.
\end{definition}

This definition implies that a picture DAG $G$
encodes a picture $P$ with dimensions $(m,n)$
as a DAG with the set of vertices
$V = \{ P(i,j) \mid i \in [m] \text{ and } j \in [n] \}$,
and its boundary picture $\hat P$ through
$V = \{ \hat P(i,j) \mid i \in [m+1]_0 \text{ and } j \in [n+1]_0 \}$.

A DAG encoding specifies the set of edges for a picture DAG.
One possible DAG encoding of a picture is illustrated in Figure~\ref{fig:enc}.
A DAG encoding determines the dependencies the automata can use while processing a picture.

\begin{definition}[Picture-to-DAG Encoding]
\label{def:encdag}
A \emph{picture-to-DAG encoding} specifies the edges for a picture DAG.
It is a function
  $\aDAG:
(\Sigmasharp)^{*,*} \cup \allpicdags
  \to \allpicdags$
%it is obtained from a picture by one of the DAG encodings defined below.
that maps its input to a picture DAG $\dagdef$
by specifying the set of edges
$E$ and the mappings $\IN$ and $\OUT$ for $G$.
The argument of the function $\aDAG$ can be a (boundary) picture (the resulting DAG will have nodes corresponding to its positions) or a picture DAG (its original set of edges will be replaced with a new set of edges).

The picture-to-DAG encoding specifies the edges either uniquely, resulting in one specific DAG $G$,
or partially, which allows for several distinct DAGs,
but the encoding maps the picture (DAG)
to an arbitrary but fixed DAG $G$ in this case.
The symbol $\aDAG$ denotes the abstract DAG encoding
without edges specified.
It stands for a specific DAG encoding defined below.
\end{definition}

Note, this definition means that the picture-to-DAG encoding
does not affect the presence of a boundary.
A picture $P$ is mapped to a picture DAG $G\in\picdags$ without boundary,
whereas a boundary picture $\hat P$ is mapped
to a picture DAG $\hat G\in\hatpicdags$ encoding additionally the boundary.
The same holds for picture DAG input: the encoding does preserve the membership in $\picdags$ or $\hatpicdags$.

Picture automata use different scanning strategies,
like row-by-row or diagonal-by-diagonal.
These scanning strategies
correspond roughly to the `wiring' of the vertices in a DAG.
However, not exactly, since scanning strategies impose
a total order on the pixels of a picture,
whereas a DAG imposes only a partial order on its vertices.
All of the following picture-to-DAG encodings specify the sets of edges that depend only on the dimensions of input pictures; therefore, we call them \emph{input-agnostic}.

\begin{definition}[Input-agnostic picture-to-DAG encodings]\label{def:blindedgesets}
The \emph{input-agnostic DAG encoding} maps its input to a unique DAG
$\dagdef$.
In case the input is a picture DAG  $G' = (V, E', \ell, \IN', \OUT')$,
the edges $E'$, as well as their ordering and wiring
specified by $\IN'$ and $\OUT'$ are replaced by the DAG encoding.
This allows for reencoding with another DAG encoding.
The following sets of \emph{input-agnostic edges} depend only
on the dimensions $(m, n)$ of the picture $P$ to be encoded,
but not on its symbols.
\\
For
$i \in [m]$,
$j \in [n]$,
$i_0 \in [m]_0$,
$j_0 \in [n]_0$,
$\hat i,\hat k \in [m+1]_0$
and
$\hat j,\hat l \in [n+1]_0$, we define a set of edges.
%as follows:
    \begin{itemize}
    \item
      $\edownleft =
      \left\{ e \mid \src(e) = \hat P(i_0,0), \tar(e) = \hat P(i_0+1,0)\right\}$
      are downward along the left border.
    \item
      $\edownright =
      \left\{ e \mid \src(e) = \hat P(i_0,n+1), \tar(e) = \hat P(i_0+1,n+1)\right\}$
      are downward the right border.
    \item
      $\hat E_\downdownarrows =
      \left\{ e \mid \src(e) = \hat P(i_0,\hat j),
                     \tar(e) = \hat P(i_0+1,\hat j)
      \right\}$
      are the vertical edges oriented downwards.
    \item
      $\ErightrightArrows =
      \left\{ e \mid \src(e) = \hat P(\hat i,j_0), \tar(e) = \hat P(\hat i,j_0+1)\right\}$
      are horizontal oriented to the right.
    \item
      $\Erfa =
      \left\{ e \mid \src(e) = \hat P(i_0,n+1), \tar(e) = \hat P(i_0+1,0)\right\}$
      are the `return' edges for the $\RFA$.
    \item
      The horizontal edges oriented alternatingly
      to the left and to the right for the $\BFA$ are
      \\
      $\begin{aligned}
      \ErightleftArrows =
      &\left\{ e \mid \src(e) = \hat P(\hat i,j_0), \tar(e) = \hat P(\hat i,j_0+1),
      i \text{ is odd } \right\}\\% right to left
      \cup %TODO spacing
      &\left\{ e \mid \tar(e) = \hat P(\hat i,j_0), \src(e) = \hat P(\hat i,j_0+1),
      i \text{ is even } \right\} % left to right
      \end{aligned}$
      \\
    %\item
      $\begin{aligned}
      \Ebfa =
      &\left\{ e \mid \src(e) = \hat P(i_0,0), \tar(e) = \hat P(i_0+1,0),
      i \text{ is odd } \right\}\\
      \cup
      &\left\{ e \mid \src(e) = \hat P(i_0,n+1), \tar(e) = \hat P(i_0+1,n+1),
      i \text{ is even} \right\}
      \text{ are vertical $\BFA$ edges.}
      \end{aligned}$% right to left
    \item
      $\hat E_{\sesearrow}=
      \left\{ e \mid \src(e) = \hat P(i_0,j_0), \tar(e) = \hat P(i_0+1,j_0+1)\right\}$
      are the diagonals to the south east.%, thus right below
      %\left\{ e \mid \src(e) = \hat P(\hat i,\hat j), \tar(e) = \hat P(\hat k,\hat l) \text{ with } \{\hat i,\hat j, \hat k,\hat l\} \cap \{0,m+1,n+1\}\neq\emptyset\right\}$
%%%    \item
%%%      $\erightrightArrows =
%%%      \left\{ e \mid \src(e) = P(i,j), \tar(e) = P(i,j+1)\right\}$
%%%      are the horizontal edges oriented to the right.
    %\item
    %  $E_{\swswarrow}=
    %  \{ e \mid \src(e) = \hat P(i,j), \tar(e) = \hat P(i+1,j-1)\}$
    %  are the diagonal neighbours to the south west.%, thus left below
    \end{itemize}
Now we define several picture-to-DAG encodings by specifying their set of edges.
We assume the corresponding $\IN$ and $\OUT$ functions are defined in a natural way
such that the edges do not cross.
Furthermore, we define
$\Eall = \left\{ e \mid \src(e) = u, \tar(e) = v, u \neq v \text{ for all } u,v\in V\right\}$
as the set of all possible edges for a picture,
in order to get rid of the edges connected to the boundary in case
the picture-to-DAG encoding maps a picture without a boundary to a DAG.
\begin{itemize}
\item
$\linlDAG$ defines
  $
  E =
  \edownleft
  \cup
  \ErightrightArrows
  \cap \Eall
  $
   joining the leftmost vertices %downwards
  and all vertices in each row.%from left to right.
\item
$\linrDAG$ defines
  $
  E =
  \ErightrightArrows
  \cup
  \edownright
  \cap \Eall
  $
   joining the rightmost vertices %downwards
  and all vertices in each row.%from left to right.
\item
$\linDAG$ defines
  $ E = \edownleft \cup \ErightrightArrows \cup \edownright \cap \Eall$
  %  joining both the leftmost and rightmost vertices downwards
  % and all vertices in a row from left to right.
   joining the left and the right border and the rows.
\item
$\rfaDAG$ defines
  $
  E =
  \ErightrightArrows
  \cup
  \hat E_{\text{rfa}\swarrow}
  \cap \Eall
  $
  as the $\RFA$ scanning strategy.
\item
$\bfaDAG$ defines
  $
  E =
  \ErightleftArrows
  \cup
  \hat E_{\downarrow\text{bfa}\downarrow}
  \cap \Eall
  $
  as the $\BFA$ scanning strategy.
%\item
%$\firDAG$ defines
%  $
%  E =
%  E_\downarrow
%  \cup
%  E_\rightleftarrows
%  $
%   joining the leftmost vertices downwards
%  and all vertices in a line with alternating orientation.
%  %from left to right.
\item
$\cooDAG$ defines
  $
  E =
  \ErightrightArrows
  \cup
  \hat E_\downdownarrows
  \cap \Eall
  $
  as a vertex pointing to two of its neighbours: below and right.
\item
$\diaDAG$ defines
  $
  E =
  %TODO
  %\cup
  \hat E_{\sesearrow}
  \cap \Eall
  $ as
  every vertex points to
  one neighbour, the one right below.
%\item
%$\diaDAG$ defines
%  $
%  E =
%  %TODO
%  %\cup
%  E_{\swswarrow}
%  $,
%  every vertex points to
%  only one neighbour, namely to the one left below.
\end{itemize}
%- TODO: We extend the mappings to picture languages and picture language families.
%- TODO: STRIPPING OFF edges
%The edges are ordered by $\IN$ and $\OUT$
%in the natural way such that they do not cross. (TODO more formally)
\end{definition}

\noindent
Instead of using the total order of a scanning strategy
or partial orders to define input-agnostic edges, we introduce
edges depending on the input symbol at the picture's positions,
called \emph{input-driven}.

% idea Frantishek: use superosition as notation for blind + driven
\begin{definition}[Input-Driven picture-to-DAG encoding]\label{def:inputdriven}
Let $G' = (V, E', \ell, \IN', \OUT')$ be a picture DAG over $\Sigma$,
either encoding the input picture $P$ or $\hat P$, with
$E' = \emptyset$ and $\IN' = \OUT' = \lambda$,
or given as input itself.
Then, the \emph{input-driven DAG encoding} maps
$G'$ to an arbitrary but fixed DAG
$\dagdefi$ \emph{over the doubly ranked alphabet} $(\Sigma$, r).
The edge set $\einputdriven$ denotes the set of \emph{input-driven edges}.
The ranks $r(\ell(v)) = (r(\ell(v))_1,r(\ell(v))_2)$
of the labels' vertices, with the vertices $v \in V$,
determine the input-driven edges:
% Here, $r(\sigma)_1$ determines the number of ingoing input-driven edges,
% and $r(\sigma)_2$ the number of outgoing input-driven edges,
% with $\sigma \in \Sigma$,
%can be constructed as follows.
%If the input-agnostic DAG encoding determined the input-agnostic edges
\begin{itemize}
\item
  $\IN(v) = \IN'(v) \cdot
  e$ with %distinct edges
  $e\in \einputdriven^*$ and $|e|=r(\sigma)_1$
  and all edges in $e$ are pairwise distinct
  and
\item
  $\OUT(v) = \OUT'(v) \cdot
  e$ with %distinct edges
  $e\in \einputdriven^*$ and $|e|=r(\sigma)_2$
  and all edges in $e$ are pairwise distinct.
  %and every edge in $e$ is distinct
  %e_1 \cdot e_2 \cdot$ \dots $\cdot e_o$
  %where $o=r(\sigma)_2$ and
  %$e_1,e_2,\dots, e_i \in \einputdriven$
\end{itemize}
The input-driven picture-to-DAG encoding is undefined
if the above construction does not yield a DAG.%
\footnote{
It is thus defined if and only if
the in- and outgoing edges are balanced:
$\sum_{v \in V} r(\ell(v))_1 = \sum_{v \in V} r(\ell(v))_2$
and the newly added input derived edges do not introduce a cycle.%
}
%$\einputdriven = \{ e \mid \src(e), \tar(e), e$
%does not introduce a cycle $\}$
%determines the in and out degree of a vertex labeld by $\ell$
%without specifying the wiring thus which edge is used.
\end{definition}
%The input-driven DAG encoding does not provide a unique DAG for a picture but a set of DAGs.
Each edge in $\einputdriven$ lacks explicit designation
of its source and target vertices,
allowing for all combinations of joining the vertices
according to the input-driven in- and outgoing specifications,
as long as the combination yields a valid DAG.
As such, this DAG encoding accommodates a variety of wiring combinations.

%%%%%     To be compatible with automata working on pictures, for an input picture $P \in \Sigma^{m,n}$, we will encode the boundary picture $\widehat{P}$ as a DAG %denoted as $\Dag(P) = (V,E)$
%%%%%     with the set of nodes $V =\{n(i,j) \mid 0 \le i \le m+1, 0 \le j \le n+1\}$. Each node $n(i,j)$ corresponds to the symbol $\widehat{P}(i,j)$ of the boundary picture $\widehat{P}$.
%%%%%     The graph
%%%%%     %$\Dag(P)$
%%%%%     has maximum in- and outdegree two,
%%%%%     its only root corresponds the upper left corner of $\widehat{P}$ and its only leaf corresponds to the bottom right corner of $\widehat{P}$. The set of edges of $\Dag(P)$ consists of edges $(n(i,j), n(i,j+1)$ for each $1 \le i \le m+2$, $0 \le j < n+1$, and $(n(i,j), n(i+1,j))$, for each $0 \le i < m+1$, $0 \le j \le n+1$.
%%%%%     For a sample DAG encoding a picture see Figure \ref{pic:enc-arrow}.

A swap operation is a fundamental operation on DAGs \cite{journals/iandc/BlumDrewes2019}. Let $G =(V, E, \ell, \IN, \OUT)$ be a DAG. Two edges $e_0,e_1 \in E$ are independent if there is no directed path from $\tar(e_i)$ to $\src(e_{1-i})$ for $i \in \{0, 1\}$. Then, the \emph{edge swap} of $e_0$ and $e_1$ yields the DAG $G' =(V, E, \ell, h \circ \IN, out)$, where $h$ is a bijection defined as follows:
$$h(e) = \left\{\begin{array}{ll}
    e_{1-i} & \mbox{if } e = e_i \mbox{ for some } i \in \{0, 1\},\\
    e       & \mbox{otherwise.}
\end{array}
\right.$$

A DAG automaton cannot be restricted to picture DAGs by its rule set.
The rules allow accepting DAGs that do not represent pictures and are obtained from DAGs encoding pictures by the swap operation.
For using DAG automata as a device recognizing picture languages,
it is necessary to always have an intersection with the picture languages $\Sigmastarstar$.

% the following definition is hard to read
%\begin{definition}[Equivalence $\eqdag$]\label{def:eqpic}
%Two automata $A$ and $A'$ are \emph{picture equivalent}
%with respect to the DAG encoding $\aDAG$, denoted as $A \eqdag A'$,
%if they recognize the same language of pictures,
%where a picture $P$ is encoded as $\aDAG(\hat P)$
%if the automaton is a DAG automaton.
%The equivalence %relation
%is lifted to language families.%
%\footnote{%
%For a picture automaton $A$ and a DAG automaton $A'$,
%this would mean $A \eqdag A'$ $\iff$
%$\aDAG(L(A)) = L(A') \cap\aDAG(\Sigmastarstar)$.}
%$$A \eqdag A' \iff
%\aDAG(L(A))  \cap\aDAG(\Sigmastarstar)
%=
%\aDAG(L(A')) \cap\aDAG(\Sigmastarstar)
%.$$
%\end{definition}

\begin{definition}[Equivalence $\eqdag$]\label{def:eqpic}
Let $A$ be a DAG automaton and $A'$ be an automaton accepting a picture language $L(A)$. The automata $A$ and $A'$ are \emph{picture equivalent} with respect to a picture-to-DAG encoding $\aDAG$, denoted as $A \eqdag A'$,
if $L(A) \cap \aDAG(\Sigmastarstar) = \aDAG(L(A'))$.
If $\mathcal{A}$ is a family of DAG automata and $\mathcal{A}'$ is a family of picture automata, the families are equivalent with respect to $\aDAG$, denoted as $\mathcal{A} \eqdag \mathcal{A}'$ if,
for each DAG language $L$, it holds
$L \in \mathcal{L}(\mathcal{A}) \cap \aDAG(\Sigmastarstar) \Leftrightarrow L \in  \aDAG(\mathcal{L}(\mathcal{A}'))$.
\end{definition}

\section{Comparison}
\label{sec:comp}

% picture automata (c.f. Sect.~\ref{subsec:piccomp}) and
% string automata (see Sect.~\ref{subsec:stringcomp})
% to DAG automata.

This section provides a comparative analysis of
string automata, designed for one-dimensional input (see Sect.~\ref{subsec:stringcomp}), in relation to their capabilities vis-à-vis DAG automata%
,
and
picture automata, which operate in two dimensions (c.f. Sect.~\ref{subsec:piccomp})%
.

\subsection{1D: Comparison of String Automata with DAG Automata}
\label{subsec:stringcomp}

% - pixel! picture symbol (not boundary) or nontransparent picture symbol??? for sublinear space
% - RESTRICT INPUT TO AUTOMATA: DAG automata run on DAGs, implications of this, restrict to other input e.g. picture dags
%Back to the one-dimensional case: %how can we pimp up the string case?
Just as DAG automata recognize the regular string languages
if the strings are encoded as string DAGs (Theorem \ref{theorem:1dDDAeqL3}),
so do picture automata $\RFA$, $\BFA$, $\twoOTA$, and $\twoDOTA$ when operating
in one dimension only (Theorem \ref{theorem:1dpicAeqFSA}).
String DAG (Def.~\ref{def:dag}) joins vertices corresponding to the letters of a string in a simple path.
But obviously, a node in a graph can have more than one incoming edge and one outgoing edge.

String DAGs and picture-to-DAG encodings from Def.~\ref{def:blindedgesets} define edges connecting nodes representing symbols in a string or in a picture in an input-agnostic fashion.
However, based on Def. \ref{def:inputdriven}, we can also add edges connecting symbols in an input-driven way. Since we use DAG automata for one- and two-dimensional strings, extra edges need to be added.

Edges are added either in an \emph{input-agnostic} fashion
%- We could use a TEMPL/2SA automaton for the input-agnostic edges
\cite{learnpicautomata}
or an \emph{input-driven} way~\cite{DBLP:journals/tcs/KutribMW21},
meaning that the edges depend either only on the dimensions of the input
or also on the input itself.
DAG encodings as defined in Def.~\ref{def:blindedgesets} are input-agnostic.

%\subsubsection{Input blind}
Adding to a string DAG input-agnostic edges computed
by $\FSA$s does not seem to add power to DAG automata.
An input-agnostic edge $e$ corresponds to the length $n$ of the substring
$v_1v_2\dots v_n$ with
$\src(e) = v_1$ and
$\tar(e) = v_n$, because if an $\FSA$ is input-agnostic, it sees a unary language
where the strings can only differ in their length.
But a $\DDA$ on string DAGs accepts the same strings as a $\FSA$.
Every additional input-agnostic edge requires an $\FSA$ for a unary language.
Since the DAG automata recognize graphs of bounded degree,
only a constant number of input-agnostic edges are allowed,
which in turn means that in each vertex, only a constant number of $\FSA$s
can start the computation of the spanning length of the input-agnostic edge.
A vertex with both in and outgoing input-driven edges
would use the final state for the unary $\FSA$ as starting state for the new unary $\FSA$.
This indicates that an $\FSA$ can,
by power set construction, keep track of all $\FSA$s computing
the input-agnostic edges.

Thus,
only a device more powerful than a $\FSA$ can exceed the regular languages
by computing input-agnostic edges for the string.
One example would be to add edges spanning $n$ vertices where $n$ is
a power of two, a square, or a prime, requiring a stack, a linear bounded tape, or an infinite tape, respectively.
%and nevertheless the DAG automata recognizing strings
A DAG automaton processing strings encoded with input-agnostic edges
computed by machines more powerful than a mere $\FSA$
exceeds thus the power of the regular languages.
Because input-agnostic edges require this trick of preprocessing
with a more powerful device,
we take a look at input-driven edges in the hope of
avoiding devices more powerful than $\FSA$s.
%TODO:
%
%- input-agnostic corresponds to unary plus dimension
%- input-agnostic grid encoding $\cooDAG$ is not $\FSA$ constructible
%
%- allowing all possible k-path DAGs for encoding string languages,
%  does encode language properties into the DAG,
%  hiding it from the membership complexity
%
%- additional edges which are input-agnostic and \textbf{computed with a $\DFA$}
%  do not change anything, still regular string languages recognized
%

%\subsubsection{Input driven}
% \input{theorem/lindagstringdag} this is not new, I do not understand why I wrote this here?!?
Input-driven edges do not require edges to be computed by a more powerful device than an FSA.
Nevertheless, the power of DAG automata recognizing strings
encoded with input-driven edges exceeds the class of regular languages. Such automata can accept, e.g., languages $L_1 = \{a^n b^n \mid n \in \mathbb{N} \}$ and $L_2= \{ a^n b^n c^n \mid n \in \mathbb{N} \}$ when input-driven edges are added so that the double rank $r$ of the symbols is $r(a)=(0,1)$, $r(b)=(1,0)$, in case of $L_1$, and $r(b)=(1,1)$, in case of $L_2$, and $r(c)=(1,0)$.

%\begin{theorem}%[$\Pi(a^nb^n) \in TODO$]
\begin{restatable}{theorem}{theoremanbn}
\label{theorem:anbn}
%\input{fig/anbn}
%\begin{theorem}%[$\Pi(a^nb^n) \in TODO$]
DAG automata can recognize both certain context-free
%\marginpar{TODO express existence: some / certain}
as well as context-sensitive languages
if strings in the string languages are encoded as
a string DAG with additional input-driven edges
determined by a ranked alphabet.
%$L = \{ w \mid \pathfirstedge(w) = a^nb^n\}$
%$a^nb^n$ string DAG plus extra magic edges recognizable
%\end{theorem}
 % #moved2appendix
\end{restatable}
%\input{proof/anbn}

%\input{theorem/anan}

%- regular languages
%  $\subsetneq$
%  input-driven
%  $\subsetneq$
%  context-free languages
%
%\input{theorem/idcounting}

\subsection{2D: Comparison of Picture Automata with DAG Automata}
\label{subsec:piccomp}
The picture automata $\RFA$ and $\BFA$ are compared to DAG automata
in Sect.~\ref{subsubsec:rfabfa},
and based on that comparison, these two models are modified
in Sect.~\ref{subsubsec:dia}.
This modification in turn leads to comparing
online tessellation automata to DAG automata in Sect.~\ref{subsubsec:ota}.

\subsubsection{Comparison of RFA and BFA to DAG automata }
\label{subsubsec:rfabfa}

Returning Finite Automata ($\RFA$) and Boustrophedon Automata (BFA)
are the most limited picture automata in literature
and recognize the same family of picture languages
(Theorem~\ref{theorem:rfaeqbfa}).

% As quite restricted automata models designed
% for processing two-dimensional inputs,
% they possess algorithmically advantageous properties
% for applications in picture recognition.
How can a DAG automaton model these two picture automata?
Scanning strategies of $\RFA$ and $\BFA$ are total orders on the positions of the picture.
Both scanning strategies are drawn
as oriented graphs in Figure~\ref{pic:stripes-rfabfa-serialized}.
% TODO: $\BFA$ intuition ox or snake or fir
% TODO: explain serialization versus $\FID$-DAG
% Let us look at the examples of stripes.
% The language $L_s$ is obviously recognized by $\RFA$ $\BFA$.
% What if we encode the pictures as DAGs?
For the $\RFA$ the big jump back from right to left was viewed critically.
Therefore, the $\BFA$ uses the alternating directions.
Let us look at an example of how DAG automata process pictures.

\begin{example}\label{ex:stripes}

Let $\Lstripes$ be the picture language of
horizontal stripes where adjacent rows are of different colors.
Formally, over a binary alphabet $\Sigma = \{\stateW,\stateB\}$, this would be
\begin{displaymath}
\Lstripes =
\left\{
\begin{array}{c}
a^n  \\
b^n  \\
a^n  \\
b^n  \\
\vdots \\
\end{array}
\ \bigg| \
\{a,b\}
=
\{\stateW,\stateB\}
\text{ and }
n\in\mathbb{N}
\ \right\}.
\qquad
\end{displaymath}

Both $\RFA$ and $\BFA$ scan the input horizontally,
to recognize horizontal lines.
Also, they can recognize stripes,
as well as the above-defined stripe language $\Lstripes$,
illustrated in Fig.~\ref{pic:stripes-rfabfa-serialized}.
And so can the DAG automaton, which simulates those two models
by applying the encodings $\rfaDAG$ and $\bfaDAG$.
This can be accomplished deterministically,
since the scanning strategy is a serialization.
The serialization of a picture yields a string,
or, in the case of a DAG automaton, a string DAG,
and apart from the scanning strategy,
both an $\RFA$ and a $\BFA$ act like an $\FSA$
for which the nondeterminism adds no power,
see Theorem~\ref{theorem:bfa-deteqnondet}.

Abandoning the total order,
we now consider encodings based on partial orders.
We examine the encodings $\linlDAG(\hat P)$ and $\linrDAG(\hat P)$,
which provide structured representations of an input picture $P$
under relaxed ordering. We begin with $\linlDAG(\hat P)$.
Consider a DAG automaton for the language $\Lstripes$
which uses the states $\stateBW$ and $\stateWB$ for encoding the alternation
between two row colors in the edge labels.
A rule cycle for the vertical edges
  ensures the alternating colors:

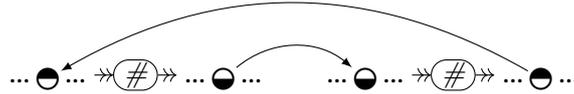
\begin{center}
~\\
               $
	       \dots
               \anchor{b4}{\stateBW}
	       \dots
               \sharparrow
	       \dots
               \anchor{b1}{\stateWB}
	       \dots
           \qquad
	       \dots
               \anchor{b2}{\stateWB}
	       \dots
               \sharparrow
	       \dots
               \anchor{b3}{\stateBW}
	       \dots
               $
               \begin{tikzpicture}[remember picture,overlay]
                       \path (b1) edge[-latex,out=40,in=140] (b2)
                             (b3) edge[-latex,out=150,in=30] (b4);
               \end{tikzpicture}

\end{center}

\newcommand\Ralternate{R_\squarebotblack}
\newcommand\rulesalternate{ 
  \Ralternate=\{
  (\stateWB \sharparrow \stateBW\stateW),
  (\stateBW \sharparrow \stateWB\stateB)\}
}
\newcommand\rulesalternateminimal{ 
  R_\squaretopblack=\{
  (\stateW \sharparrow \stateB\stateW),
  (\stateB \sharparrow \stateW\stateB)\}
}
\newcommand\Rstripes{R_\Equiv}
\newcommand\rulesstripes{ 
  \Rstripes = \{
  (\stateB \Barrow \stateB),
  (\stateW \Warrow \stateW),
  (\stateB \Barrow \lambda),
  (\stateW \Warrow \lambda)\}
%<  (\stateB \sharparrow \#),
%<  (\stateW \sharparrow \#),
}
\newcommand\Rborder{R_\#}
\newcommand\rulesborder{ 
  \Rborder = \{
  (\stateWB \sharparrow \#),
  (\stateBW \sharparrow \#),
  (\# \sharparrow \#),
  (\# \sharparrow \lambda)\}
}
\newcommand\Rroot{R_{root}}
\newcommand\rulesroot{ 
  R_{root} = \{
  (\lambda \sharparrow \stateBW\#),
  (\lambda \sharparrow \stateWB\#)\}
}
\newcommand\rulesrootmin{ 
  R_{root} = \{
  (\lambda \sharparrow \stateB\#),
  (\lambda \sharparrow \stateW\#)\}
}

A DAG automaton
$(
\{\stateBW,\stateWB,\stateB,\stateW,\#\},
\{\stateB,\stateW\},
\Rroot\!\cup\Ralternate\cup\Rstripes\cup\Rborder
)$
for the encoding $\linlDAG$ shown in Fig.~\ref{pic:stripes-grass-left}
could label the edge path on the left border
with the states $\stateWB$ and $\stateBW$
by a rule cycle alternating the rules $\rulesalternate$.
In this way, it can alternate between the colors on the left border
and accept the lines in case they comply with the pattern memorized
with the help of the states, vertical edge labels, at the left border.
Additionally, it needs the rules $\rulesstripes$
for accepting the colored lines themselves,
and the rules $\rulesborder$ for $\hat P$'s lower border.
The rules specified so far are top-down deterministic, but lack
a specification for a root.
The rule set of a nondeterministic DAG automaton may contain both
rules in $\rulesroot$,
which allows the automaton to recognize $\Lstripes$.
In contrast, a $\DDA$ lacks the ability to guess the initial color.
Since, due to its determinism, a $\DDA$ is restricted to choosing
one of the two root rules in $\Rroot$ for its rule set.
It cannot include both of them.
But this hardcodes the color.
Consequently, a deterministic DAG automaton using encoding $\linlDAG$ can only recognize
  the stripe language starting with a fixed color.
This is due to the left border being blind to the colors of the stripes
when processing the picture top-down deterministically.
The stripes depend on the left border but not vice versa.
 A $\DDA$ cannot recognize $\Lstripes$ encoded as $\linlDAG$.

However, with the right border vertically connected as in
Fig.~\ref{pic:stripes-grass-right},
the vertical edges depend on the colors of the lines
in a top-down deterministic run.
Consequently, a $\DDA$ can recognize all pictures
$\linrDAG(\hat P) \in \Lstripes$ rather than only those pictures $P$
with a fixed initial color.
The rules
    $(\stateB\stateWB \sharparrow \stateBW)$
       and
    $(\stateW\stateBW \sharparrow \stateWB)$
assign the edge labels deterministically downward,
thereby propagating the color information.

\end{example}

Let us now formalize and generalize the observations
that we have made on $\Lstripes$.

% string / total order / no poset
%%%%\input{theorem/stringdag2string}
%%%%\input{theorem/fideqfd4string}
\begin{restatable}[$\DDA\eqrfa\RFA$, $\DDA\eqbfa\BFA$] {theorem}{theoremstringdagrfabfa}
%\begin{theorem}[$\DDA\eqrfa\RFA$, $\DDA\eqbfa\BFA$]
\label{theorem:stringdagrfabfa}
  $\RFA$s and $\BFA$s recognize the same picture languages as $\DDA$s
run on pictures encoded by the scanning strategy of $\RFA$ and $\BFA$,
$\rfaDAG$ and $\bfaDAG$,
respectively.
% Thus:
% $\DDA\eqrfa\RFA$
% and
% $\DDA\eqbfa\BFA$.

\end{restatable}

Contrary to the total order on the pictures' pixels
(resulting in ordinary string languages),
when we use DAG encodings inducing a proper partial order
(\emph{proper} meaning no total order),
nondeterminism can equip a DAG automaton with more power.

% outcommented temporarilly for NCMA due to space and time constraint

\begin{restatable}[$\ddalinr$ simulates $\ddalinl$] {lemma}{lemmaddalinrsimddalinl}
%\begin{lemma}[$\ddalinr$ simulates $\ddalinl$]
\label{theorem:ddalinrsimddalinl}
  %\begin{lemma}[$\ddalinr$ simulates $\ddalinl$]
A deterministic DAG automaton can simulate using the DAG-encoding 
$\linlDAG$ for a picture $P$ 
by encoding it as $\linrDAG(\hat P)$, thus
$\linlDAGL(\DDA)\subseteq\linrDAGL(\DDA)$.
%$\DDA_\symbollinr$
%can simulate $\DDA_\linlDAG$
%\end{lemma}
 % #moved2appendix
\end{restatable}

\begin{restatable}[$\linlDAGL(\DDA)\subsetneq\linrDAGL(\DDA)$] {theorem}{theoremlinlDAGddasubsetlinrDAGdda}
%\begin{theorem}[$\linlDAGL(\DDA)\subsetneq\linrDAGL(\DDA)$]
\label{theorem:linlDAGdda-subset-linrDAGdda}
  %\begin{theorem}[$\linlDAGL(\DDA)\subsetneq\linrDAGL(\DDA)$]
The class of picture languages recognized by deterministic DAG automata
with input pictures $P$ encoded as $\linlDAG(\hat P)$
is a proper subclass of those encoded as $\linrDAG(\hat P)$.
%, thus\\
%$\linlDAGL(\DDA)\subset\linrDAGL(\DDA)$.
%\end{theorem}

\end{restatable}
\begin{corollary}\label{cor:noclosurelinldet}
$\DDA$ for the DAG-encoding $\linlDAG$ is not closed under union.
\end{corollary}
%%%%% \input{theorem/determinizematrixgrammar}

%%%%% In the nondeterministic version the DAG automata
%%%%%  are invariant to a variety of DAG encodings.
%%%%% This invariance holds if the DAG-encodings
%%%%% connect all their vertices in a row
%%%%%  and every row points via an edge to its row below.
%%%%%
%%%%% \input{theorem/eqlinlDAGnda2rfa}
%%%%% %\input{theorem/eqlinrDAGnda2rfa}
%%%%% \input{theorem/eqfirdag2bfa}
%%%%% \input{theorem/eqrowsrandom1down2rfa}
%%%%% % \input{theorem/eqtotalorderrfabfabyfid}
%%%%% Does it help to have both the right and the left border vertically connected?
%%%%% \input{theorem/eqlindag2rfa}

The scanning strategy seems to determine,
which languages an $\RFA$ or, equivalently, a $\BFA$ can recognize,
which would mean that the language is partly
encoded within the scanning strategy itself.
This leads us to the discussion on the modification of $\RFA$s and $\BFA$s,
which allows for greater flexibility and capability
in recognizing a wider range of languages.

\begin{figure*}[]
\hfill
\begin{subfigure}{0.3\textwidth}\center
  \input{pic/rfabfa-stringdag}
% The serialization of a picture as a string DAG,
% for the $\RFA$ (green) and $\BFA$ (brown) scanning strategy,
% enables a deterministic DAG automaton
% to recognize alternating stripes.
  \caption{The serialization of a picture as a string DAG
for the $\RFA$ (green) and $\BFA$ (brown)
enables deterministic recognition of the language $\Lstripes$.\\
  }
  \label{pic:stripes-rfabfa-serialized}
\end{subfigure}
\hfill
\begin{subfigure}{0.3\textwidth}\center
    \input{pic/grass-left}
    \caption{For $\linlDAG$, with vertical edges only at the left border,
          deterministically cannot be detected
	  which stripe color occurs first; the first color is hard-coded \dots}
    \label{pic:stripes-grass-left}
\end{subfigure}
\hfill
\begin{subfigure}{0.3\textwidth}\center
    \input{pic/grass-right}
    \caption{\dots{} however this is possible
          with the right border connected, instead.
	  For $\linrDAG$, all stripe DAGs are accepted,
	  not only those with the fixed first color.}
    \label{pic:stripes-grass-right}
\end{subfigure}
\captionsetup{subrefformat=parens}
\caption{
The DAG encodings
$\bfaDAG$,
$\rfaDAG$,
$\linlDAG$ and
$\linrDAG$
recognize horizontal stripes \emph{deterministically}:
\\
The above figures illustrate with their red edges those edges
where the first row's color is determined.
\subref{pic:stripes-rfabfa-serialized}~While
the total orders on the vertices imposed by the $\RFA$ and the $\BFA$
scanning strategies with the DAG encodings
$\bfaDAG$ and $\rfaDAG$ can detect the first row's color,
\subref{pic:stripes-grass-left}~the DAG encoding
$\linlDAG$, even though permitting the same number of edges,
is blind to the row colors in the input.
Therefore, even though a DAG automaton can detect striped pictures with
the encoding $\linlDAG$, it cannot recognize $\Lstripes$,
but only a subset of it -- the set of those pictures whose first line
exhibits the hard-coded color.
\subref{pic:stripes-grass-right}~However, sees the first row's color
and can alternate the colors subsequently.
}
\label{fig:stripes}
\end{figure*}

\subsubsection{Adapted Scanning Direction for $\RFA$ and $\BFA$}
\label{subsubsec:dia}

As hinted in the previous section, an
$\RFA$ can recognize different languages by altering its scanning strategy. More precisely, it need not adapt the strategy; changing direction suffices.
 
%As hinted in the previous section, an
%\RFA can recognize different languages by altering its scanning strategy. More precisely, it does not need to adapt the strategy; it suffices to change direction.
\begin{example}[\tbd{DAG language of diagonals}]\label{ex:dia}
Let $\Ldia$ be the language where on every diagonal line
oriented southeast all pixels have the same color.
Fig.~\ref{fig:cycdia} illustrates this language with one diagonal only.
%Formally,
$$\Ldia =  \{\; P \in \Sigma^{m,n}\mid P_{i, j} = P_{i+k, j+k}
\text{ for } i,i+k\in[m]
\text{ and } j,j+k\in[n]
\;\}$$%.
The language $\Ldia$ is not recognizable by an
RFA~\cite[Lemma 47]{DBLP:journals/jcss/FernauPST18}\footnote{The language
in the cited proof is restricted to a specific diagonal
but that does not affect the validity of the pumping argument
for our language.}.

A DAG automaton can recognize a picture $\Pdia \in \Ldia$ deterministically
by encoding $\Pdia$ as $\diaDAG(\hPdia)$.
The rules
$\lambda\sharparrow\#$
and
$\#\sharparrow\lambda$
detect the border.
For all colors $\sigma\in\Sigma$, the rules
    $\#\sigmaarrow\sigma$,
    $\sigma\sigmaarrow\sigma$ and
    $\sigma\sigmaarrow\#$
ensure one fixed color for each diagonal.
Like that, a $\DDA$ accepts a DAG $\diaDAG(\hPdia)$.
By Def.~\ref{def:encdag}, $\diaDAG(\hPdia)$ is a disconnected DAG.
The DAG encoding $\diaDAG$ uses the edges
illustrated in green in Fig.~\ref{pic:dia}.
We could add,
as illustrated with black arrows in Fig.~\ref{pic:dia},
the edge sets $\edownright$
and edges from right to left at the bottom border to the DAG encoding,
in order to obtain a \emph{connected} DAG.
But we do not need these edges for recognizing the language $\Ldia$.
The automaton would need it however,
for recognizing diagonal stripes.
Recall that the border ensured the stripes in Example~\ref{ex:stripes}.
Analogously, with these additional edges,
a DAG automaton recognizes diagonal stripes deterministically.
\end{example}

An $\RFA$ scans row by row.
In order for an $\RFA$ to recognize $\Ldia$,
we could either rotate the pictures by the angle of $45^\circ$ clockwise,
or we could adapt the direction of the scanning strategy.
If an $\RFA$ uses its scanning strategy in the direction southeast
instead of east,
thus $\measanglerdtose$,
it will scan diagonals instead of rows. Similarly,
a $\BFA$ can alternate its scanning direction
between northwest and southeast instead of west and east.
The \emph{total order} for the diagonal scanning of an $\RFA$ %strategy
of a two-dimensional string $S$ in the direction southeast
is given by:
%\begin{itemize}
%\item $\RFAdia$:
$$ S(i_1j_1) \prec S(i_2j_2)
\emph{ if }
  (i_1+i_2 < j_1+j_2)
  \lor ((i_1+i_2 = j_1+j_2) \land (j_1 < j_2))$$
%  and
%$(i < k) \lor (i = k \land j < l)$
%\item $\BFAdia$:
%$S(i_1j_1) < S(i_2j_2) $
%if
%$
%  (i_1+i_2 < j_1+j_2)
%  \lor
%  ((i_1+i_2 = j_1+j_2) \land
%  \\
%    ((i_1 \mod 2 = 0 \land j_1 < j_2) \lor
%     (i_1 \mod 2 = 1 \land j_1 > j_2)))
%$
%.
%\end{itemize}
and similarly for a $\BFA$.

This shows that an $\RFA$ can detect lines in the scanning direction.
With a DAG encoding $\squarevfill$ or a scanning direction downwards,
$\DDA$ or $\RFA$, respectively,
    could recognize a language with vertical lines $L_|$.
With a DAG encoding $\symbolDia$ or a scanning direction southwest,
$\DDA$ or $\RFA$ could recognize
a language with diagonal lines oriented to the southwest $L_/$.
Adapting the scanning strategy to the language is kind of `cheating'
because we should not know which language we will parse.
Otherwise,
we could always encode a part of the language into the scanning strategy.
This would outsource the complexity of the membership problem for a language
into its automaton's description via the scanning strategy.
So we would hide a part of the language in the scanning strategy.
This can be appropriate for certain use cases.
However, if we do not know which angle for lines to expect in advance, we can use a grid.
A grid has the advantage that it is not tailored to the language.
The DAG encoding $\cooDAG$ serves as one option for implementing a grid.
Given a picture DAG $\cooDAG(\hPdia)$
with one diagonal of the color $\stateB\in\Sigma$,
the rule for a pixel of the diagonal
$tq\Barrow t\!p$
and the rule for the pixel on its right-hand side
$pt\Warrow qt$
form the rule cycle in Fig.~\ref{pic:diarulecycle}.
This rule cycle yields the diagonal in the picture
shown in Fig.~\ref{pic:diacyclepath}.
%The DAG encoding $\cooDAG$ can detect straight lines for all angles.
It is easy to verify that the DAG encoding $\cooDAG$
can detect deterministically horizontal, vertical, and diagonal lines.

The DAG encoding $\cooDAG$ enables the detection of horizontal, vertical, and diagonal patterns within a single language. In contrast, an $\RFA$ or $\BFA$ must be tailored with a specific scanning direction and, as a result, can recognize only one such pattern per language. Consequently, a DAG automaton operating on pictures encoded via $\cooDAG$ exhibits strictly greater expressive power than both $\RFA$ and $\BFA$. The following section identifies the classical picture automaton model to which the picture-to-DAG encoding $\cooDAG$ corresponds.
%yvo: no idea what I wanted to say here:
%Whereas the regular string languages are closed under reversal

\begin{figure*}[t]
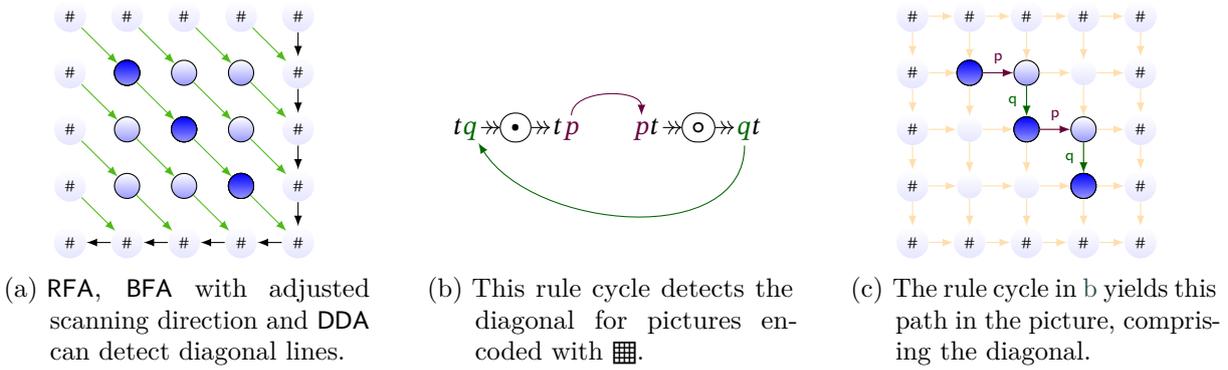

  \begin{subfigure}{0.3\textwidth}\center
    \input{pic/dia.tex}
    \caption{
    $\RFA$, $\BFA$ with adjusted scanning direction
    and $\DDA$ can detect diagonal lines.}
    \label{pic:dia}
  \end{subfigure}
  \hfill
  \begin{subfigure}{0.3\textwidth}\center
    \input{cycle/dia}
    \caption{
    This rule cycle detects the diagonal
    for pictures encoded with $\cooDAG$.
    %Note that it is not top-down deterministic.
    }
    \label{pic:diarulecycle}
  \end{subfigure}
  \hfill
  \begin{subfigure}{0.3\textwidth}\center
    \input{pic/diacyclepath}
    \caption{
    The rule cycle in
    \subref{pic:diarulecycle}
    yields this path in the picture,
    comprising the diagonal.
    }
    \label{pic:diacyclepath}
 \end{subfigure}
\captionsetup{subrefformat=parens}
\caption{
% The picture language automaton \subref{pic:adia}
% accepts DAG pictures with a diagonal \subref{pic:dia}.
% The rule cycle \subref{pic:diacyclepath}
% yields a path in the DAG \subref{pic:diarulecycle}
% for the visible diagonal
% and the transparent diagonal above it.
% The subgraph for the visible diagonal is an unconnected graph
% and thus no path in the DAG
\subref{pic:dia}
    An $\RFA$ or $\BFA$ can accept a picture $P_\backslash$
    with diagonal lines
    if the angle for the scanning strategy is decremented by $45^\circ$
    (the direction of the scanning strategy is shown in the picture,
    but not the strategy itself),
    just as a $\DDA$ can for the picture
    being encoded as $\diaDAG(\hPdia)$ (see picture).
    For the DAG encoding $\cooDAG$, which has the advantage that
    this encoding is not tailored to the language,
    the rule cycled \subref{pic:diarulecycle} yields
    the path which includes the diagonal~ \subref{pic:diacyclepath}.
}
\label{fig:cycdia}
\end{figure*}

%yvo: no idea what I wanted to say here:
%Whereas the regular string languages are closed under reversal

\subsubsection{Comparison of online tessellation automata to DAG automata}
\label{subsubsec:ota}
The DAG encoding $\cooDAG$ provides the capability to recognize 2D patterns
that an $\RFA$ cannot.
This leads us to the question: What capabilities
does a DAG automaton possess when processing pictures encoded
via the picture-to-DAG encoding $\cooDAG$?
In order to answer this question, we first show that the boundary
of a picture does not increase the expressiveness of a DAG automaton.
This observation will ease the proof of expressiveness
for the encoding $\cooDAG$.

The concept of encoding a picture $P$ as a boundary picture $\hat P$
raises the question
of whether the presence of this boundary influences expressive power.
The following lemma asserts that, for a picture language $L$, padding pictures with sentinels $\#$,
when encoding pictures using the picture-to-DAG encoding $\cooDAG$,
does not affect whether an $\NDA$ can accept the picture language $L$.

\begin{restatable}[Boundary invariance of $\NDA$] {lemma}{theoremhatPeqP}
%\begin{lemma}[Boundary invariance of $\NDA$]
\label{theorem:hatPeqP}
  Let $\hat A=(\hat N,\Sigmasharp,\hat R)$ be an $\NDA$ for recognizing
boundary pictures $\hat P\in\hatpicdags$ encoded as $\cooDAG(\hat P)$
with $P \in \picdags$.
Then there exists an equivalent $\NDA$ $A=(N,\Sigma,R)$
recognizing the same set of pictures encoded without the boundary
as $\cooDAG(P)$, but restricted to \emph{proper} pictures,
meaning pictures with more than one row and column:
$$
  \forall m,n\!>\!1\ \forall P \in\Sigma^{m,n}:
  \cooDAG(\hat P)\in L(\hat A) \iff \cooDAG(P) \in L(A).$$

\end{restatable}
In fact, the use of the
%picture-to-DAG
encoding $\cooDAG$ corresponds
directly with an online tessellation automaton:

\begin{restatable}[$\twoOTA  \eqcoo \NDA$] {theorem}{theoremnotaeqda}
%\begin{theorem}[%2DOTA det. DAG A eqivalent on pictures]
%$\cooDAG(\Lclass{2DOTA}) = \Lclass{DagA^{det}} \cap \cooDAG(\Sigmastarstar)$
\label{theorem:notaeqda}
  %\input{fig/dotaeqdetdaga4pics}
%\begin{theorem}[%2DOTA det. DAG A eqivalent on pictures]
%$\twoOTA  \eqcoo \NDA$]
%$\cooDAG(\Lclass{2DOTA}) = \Lclass{DagA^{det}} \cap \cooDAG(\Sigmastarstar)$
%\label{theorem:otaeqda}
%DAG automata run on the picture DAGs $\cooDAG(\Sigmastarstar)$ and
%two-dimenstional online tesselation automata
%recognize the same class of picture languages,
%  respectively for the deterministic and nondeterministic variants.
Online tessellation automata and DAG automata are equivalent with respect to the picture-to-DAG encoding $\cooDAG$.
%$$\begin{array}{l}
%\twoDOTA \eqcoo \DDA,\\
%\twoOTA  \eqcoo \NDA.
%\end{array}$$
%\end{theorem}

\end{restatable}

If we use input-driven DAG encoding instead of input-agnostic DAG encoding, a DAG automaton can recognize more picture languages than online tessellation automata.

\begin{restatable}{theorem}{theoremescapeotasgraffito}
\label{theorem:escapeotasgraffito}
  %\begin{theorem}[]\label{theorem:escapeotasgraffito}
There exists a picture language %$L\in\Sigmastarstar$
that a $\DDA$ recognizes with the help of an input-driven DAG encoding,
which neither an online tessellation nor a sgraffito%
\footnote{We do not provide the definition of a sgraffito automaton here,
since we only refer to the result that $\Lbalance$ cannot be recognized
by a sgraffito automaton.}%
\cite{DBLP:conf/ncma/OttoM16, DBLP:journals/ita/PrusaMO14}
automaton can recognize.
%\end{theorem}

\end{restatable}

The DAG vertex and edge labels do not reflect the topology of the picture.
We can retrieve the picture by the unique names of their vertices
$v\in \{P(i,j)$ for $i\in [m],j\in [n]\}$,
but if we prefer a topology preserved also by the wiring of the DAG,
we can combine an input-agnostic with an input-driven DAG encoding.
Encoding a picture $P\in\Lbalance$ as $\cooDAG(\hat P)$
with additional input-driven edges, as specified in the Proof of
Theorem~\ref{theorem:escapeotasgraffito} above,
yields proper DAGs for encoding
this picture language while preserving the picture not only via the unique
vertex naming but also via the topology given by the DAG encoding $\cooDAG$.
Note that $\Lbalance$ includes the picture
$P'$, where $\rfaDAG(\hat P') = b^na^n$.
Using the rule set $R$ from the theorem above
does not suffice since it would introduce cycles for $P'$.
This can be fixed easily by adding
$(\lambda \rarrow b p')$ and
$(p' \rarrow a \lambda )$ to the rules.
However, these additional rules turn the $\DDA$ into an $\NDA$
since they introduce nondeterminism.
Hence, the topology $\cooDAG$ within the edge wiring
comes at the price of losing determinism.

\section{Conclusions}
\label{sec:conclusions}

We have proposed a unifying framework for recognizing picture languages via DAG automata by encoding two-dimensional inputs into directed acyclic graphs (DAGs). Using different input-agnostic encodings of pictures into DAGs, we obtained DAG automata equivalent to two-dimensional returning finite automata, boustrophedon automata, and online tessellation automata.

%When we restrict to string languages, as languages of one-row pictures, input-agnostic encoding characterizes regular languages.\marginpar{input-agnostic edges restrict DAG automata to regular languages only if the edges are computed with a $\FSA$.}
%
%On the other hand,
Encoding pictures as DAGs using input-driven edges increases the power of DAG automata for recognizing both string and picture languages. Input-driven edges enable recognition of some context-sensitive string languages. In two dimensions, input-driven DAG automata accept a proper super-class of the class of recognizable languages.

DAG automata have many appealing properties. Blum and Drewes \cite{journals/iandc/BlumDrewes2019} proved that emptiness and finiteness are decidable for DAG automata, and the class of DAGs accepted by DAG automata is closed under union and intersection. These properties carry over to the case of accepting encodings of picture languages.
% does decidability of emptiness and finiteness really carry over to accepting picture languages with DAG automata. If the decision is that the automaton accepts a non-empty DAG, the DAG need not to encode a picture
% Similarly, if the DAG automaton accepts an infinite language, it is possible that the set of accepted encodings of pictures is not infinite.

Deterministic DAG automata can be reduced and their equivalence is decidable in polynomial time \cite{journals/iandc/BlumDrewes2019}. Equivalence for deterministic DAG automata with respect to picture language recognition is still an open problem, as two deterministic DAG automata accepting the same encodings of pictures need not be equivalent, as they can accept different sets of DAGs that do not encode pictures.

Ad Meeres \cite{AdMeeres2024} provides a framework to port known efficient finite state automata algorithms and properties from strings to DAG languages. By applying the techniques, we could obtain classes of picture languages that can be efficiently recognized by finite state automata for DAG languages.

\bibliography{bibliography_no_url}
\newpage
\appendix
\section{Appendix}
\label{sec:proofappendix}
%\noindent{\bf Proof of Theorem~\ref{} (cont.)}
%\noindent{\bf Proof of Theorem~2.8}

Here are the proofs omitted in the main text.

%\noindent{\bf Proof of Theorem~\ref{theorem:1dDDAeqL3}}
\theoremdDDAeqL*
\begin{proof}
%Let $\Sigma$ be an alphabet such that $\#\not\in\Sigma$. Then the set of all strings $w\in\Sigmastar$ such that $w$ encoded as a string DAG is accepted by a given DAG automaton is regular and for each regular language $L$, there exists a DAG automaton such that the set of
Let $A=(Q,\Sigma,R)$ be a DAG automaton. We will construct an  $\NFA$ $A_{\mathrm{nfa}}=(Q\cup\{q_0,q_f\},\Sigma,\delta,\{q_0\},$ $\{q_f\})$, where $q_0,q_f$ are new states not in $Q$, $q_0$ is the only initial state of $A_{\mathrm{nfa}}$, and $q_f$ is the only accepting state of $A_{\mathrm{nfa}}$, such that $L(A_{\mathrm{nfa}})$ is the set of all strings of the form $w$ such that $w$ encoded as string DAG is accepted by $A$. The automaton $A_{\mathrm{nfa}}$ will have the transition $(q,a,q')\in\delta$ for some $q,q' \in Q\cup\{q_0,q_f\}$ and $a \in \Sigma$ if
$$\begin{array}{ll}
   (q\circled{$a$}q')\in R & \mbox{where } q,q' \in Q \mbox{ and } q \ne \lambda \ne q'\mbox{, or} \\
   (\lambda\circled{$a$}q)\in R, & \mbox{where } q'\in Q \mbox{ and } q=q_0\mbox{, or} \\
   (q\circled{$a$}\lambda)\in R, & \mbox{where } q \in Q \mbox{ and } q'=q_f.
\end{array}$$
It is easy to show that for each string $w \in \Sigmastar$ (e.g., by induction on the length of $w$), it holds: there exists an accepting run of $A$ on a string DAG corresponding to $w$ if and only if $A_{\mathrm{nfa}}$ accepts the word $w$.
As  $\NFA$s accept the class of regular languages, it holds $\mathcal{L}_{\mathrm{onerow}}(DAG) \subseteq Reg$.

For the opposite direction, if $L$ is a regular language, then it is accepted by a $\DFA$ $A_{\mathrm{dfa}}=(Q,\Sigma,\delta,q_0,F)$. Now, we can construct a (top-down deterministic) DAG automaton $A=(Q,\Sigma,R)$ accepting the set of string DAGs corresponding exactly to words in $L=L(A_{\mathrm{dfa}})$. The automaton $A$ will have the following set of rules:
$$\begin{array}{ll}
(q\circled{$a$}q') & \mbox{for all states } q,q' \in Q \mbox{ such that } (q,a,q')\in\delta,\\
(\lambda \circled{$a$}q') & \mbox{for each state } q' \in Q \mbox{ such that } (q_0,a,q')\in\delta, and\\
(q\circled{$a$}\lambda) & \mbox{for each state } q \in Q \mbox{ such that } (q,a,q_f)\in\delta \mbox{ for some } q_f \in F.
\end{array}
$$

For a word $w\in\Sigma^*$, if the $\DFA$ $A_{\mathrm{dfa}}$ has an accepting computation on $w$, then the DAG automaton $A$ has a run on the string DAG corresponding to $w$. Hence, it holds $\mathcal{L}_{\mathrm{onerow}}(DAG) \supseteq Reg$, and the statement of the theorem follows.
\end{proof}

%\noindent{\bf Proof of Theorem~\ref{theorem:anbn}}
\theoremanbn*
\begin{proof}
The context-free language
$\{ \mathtt{a}^n \mathtt{b}^n \mid n \in \mathbb{N} \}$
can be recognized by a $\DDA$ over a ranked alphabet
$(\{\mathtt{a},\mathtt{b}\}, r)$
if it is encoded as a string DAG with additional input-driven edges.
Let $A = (\{a,b,p\},(\{\mathtt{a},\mathtt{b}\}, r),R)$ be a $\DDA$
over a ranked alphabet $(\{\mathtt{a},\mathtt{b}\}, r)$
with the rules
$$
R =
(\lambda \rarrow{a} a),
      (a \rarrow{a} ap),
     (ap \rarrow{b} b),
     (bp \rarrow{b} b),
     (ap \rarrow{b} \lambda),
     (bp \rarrow{b} \lambda).
$$
Then, $A$ will accept exactly those strings encoded as string DAGs
of the form $\mathtt{a}^n\mathtt{b}^n$
with additional input-driven edges.
The edge labels $a$ and $b$ ensure that the string is of
the form $\mathtt{a}^*\mathtt{b}^*$.
Each $\mathtt{a}$ is matched with a $\mathtt{b}$ by the edge label $p$.
Without the ranked alphabet constraints, a string with an unbalanced number of
the symbols $\mathtt{a}$ and $\mathtt{b}$
would still yield a string DAG.
However, due to the additional edges
caused by the ranked alphabet $(\Sigma ,r)$,
where $r(\mathtt{a})=(0,1)$ and $r(\mathtt{b})=(1,0)$,
only strings with an equal number of
the symbols $\mathtt{a}$ and $\mathtt{b}$ result in a valid DAG encoding.

A $\DDA$ recognizes the context-sensitive language
$\{ \mathtt{a}^n\mathtt{b}^n\mathtt{c}^n \mid n \in \mathbb{N} \}$
in the same way.
Again, the input-driven edges match the number of occurrences,
this time for three symbols.
One possible rule set achieving the desired matching would be
$$
R =
(\lambda \rarrow{a} ap),
      (a \rarrow{a} ap),
     (ap \rarrow{b} bq),
     (bp \rarrow{b} bq),
     (bq \rarrow{c} c),
     (cq \rarrow{c} c),
     (cq \rarrow{c} \lambda).
$$
Here, the vertices labeled by $\mathtt{b}$ match with
both $\mathtt{a}$, via the edge label $p$,
and $\mathtt{c}$, via the edge label $q$.
Again, the edges of the string DAG ensure
the sequence of the symbols, here $\mathtt{a}^*\mathtt{b}^*\mathtt{c}^*$,
and the input-driven edges enforce the
balance of occurrences because otherwise the string does not yield a DAG.
%TODO: See picture \ref{fig:anbn}

% TODO:
%
% - identify the edges for the string: order of edges
%
% - strip off the non-string edges: order of edges
%
% - dag of the first dimension = the string DAG $a^*b^*$
%
% - projected to first dimension = the string $a^*b^*$
%
% - dag of the second dimension = the disconnected DAG $ab \& ab \& \dots ab$
%
% - projected to the second direction = the multiset of strings $ab$

\end{proof}

%\noindent{\bf Proof of Theorem~\ref{theorem:stringdag2rfabfa}}
\theoremstringdagrfabfa*
\begin{proof}
% Proof in bullet points:
% 1. $\RFA$ $\BFA$ scanning strategy
% 2. string DAG
% 3. string DAG eq DFA
% 4. DFA eq $\RFA$ $\BFA$
Both $\RFA$ and $\BFA$ are $\FSA$s
equipped with their respective scanning strategy
for two-dimensional input yielding a one-dimensional string.
The DAG encodings mimicking these scanning strategies
by representing the picture $P$ as the DAGs
$\rfaDAG(\hat P)$ and $\bfaDAG(\hat P)$, resp.,
yield picture DAGs which are string DAGs,
illustrated in Fig.~\ref{pic:stripes-rfabfa-serialized}.
These string DAGs represent exactly the strings
the respective scanning strategies for $\RFA$ and $\BFA$ yield.
Thus:
\newcommand\directprec{\prec\!\!\prec}
$$
\src\bigl( \hat P(i_1,j_1) \bigr) = \tar\bigl( \hat P(i_2,j_2) \bigr)
\iff
\hat P(i_1,j_1) \directprec \hat P(i_2,j_2)
$$
where $\directprec$ denotes the covering relation,
meaning that no intermediate position $\hat P(i',j')$ exists such that
$\hat P(i_1,j_1) \prec\hat P(i',j') \prec \hat P(i_2,j_2)$.
According to Lemma~\ref{theorem:1dDDAeqL3}, the $\DDA$ for string DAGs
is equivalent to a $\DFA$,
just as a $\BFA$ with its scanning strategy is.
The theorem is proved since all three automaton models act on the same (underlying) string.
%completing the proof.
%So, serializing the picture into a string yields ...
\end{proof}

%\noindent{\bf Proof of Theorem~\ref{theorem:ddalinrsimddalinl}}
\lemmaddalinrsimddalinl*
\begin{proof}
For an $m \times n$ picture $P$ encoded as
$\linlDAG(\hat P)$
as well as encoded as
$\linrDAG(\hat P)$,
every row excluding the border vertices corresponds to a string DAG.
For a picture encoded by $\linlDAG$,
the rules applied on the left border
$\lambda\sharparrow q_{\#1}  q_0$,
$q_{\#0}\sharparrow q_{\#2}  q_1$,
$q_{\#1}\sharparrow q_{\#3}  q_2$, \dots
$q_{\#m}\sharparrow q_{\#m+1}q_m$,
$q_{\#m+1}\sharparrow        q_{m+1}$
assign to every $i$-th row
its individual starting state $q_i$
(while assigning the edge labels $q_{\#1}, q_{\#2}, \dots, q_{\#m}$
to the edges from $\edownleft$ on the left border).

On the contrary, for the DAG-encoding $\linrDAG$,
every string DAG for a row starts with the same start state,
let it be $q_\text{row}$,
since a deterministic DAG automaton does allow only one rule
$\lambda \sharparrow q_\text{row}$ for a vertex labeled with $\#$,
which has exactly one edge and that one being outgoing.
However, the edges $\edownright$ pointing downwards the right border
can be labeled identically to the labels on the left border with the DAG-encoding $\linlDAG(\hat P)$,
as $q_{\#1}, q_{\#2}, \dots, q_{\#m}$.
Namely, the $\DDA$ obtains the information of the respective starting state
at the end of processing the string DAG of a row from the last node (border) of the above row.
Guessing the start state can be done nondeterministically.
This does not impose a problem for a $\DFA$ processing a string,
since for regular string languages this can be determinized by
the classical power set construction.
A $\DDA$ can recognize regular string languages
(c.f.~Theorems~\ref{theorem:1dDDAeqL3}),
consequently essentially processing as a $\DFA$,
thus it can defer the information of the start state~$q_i$
for the string DAG in row~$i$ to the end of the string.
This completes the proof.
\end{proof}

%\input{theorem/linlDAGdda-neq-linrDAGdda} empty

%\noindent{\bf Proof of Theorem~\ref{theorem:linlDAGdda-subset-linrDAGdda}}
\theoremlinlDAGddasubsetlinrDAGdda*
\begin{proof}
According to Theorem~\ref{theorem:ddalinrsimddalinl}, a $\DDA$ can simulate
the DAG-encoding $\linlDAG$ by using the encoding $\linrDAG$ instead.
However, the reverse direction does not apply.
Recall that Example \ref{ex:stripes} provided the intuition on the fact that
the language $\Lstripes$ cannot be recognized by a $\DDA$ if
%the picture is encoded by
it uses $\linlDAG$, but it can when it applies $\linrDAG$.
% The picture languages stripes cannot be recognized by a $\DDA$ but by a $\NDA$
% if the picture is encoded by $\linr$.

Formally, this is due to the different partial orders on the vertices
imposed by the picture-to-DAG encodings $\linlDAG$ and $\linrDAG$.
The first dark and the first light vertices
have as their least common ancestor the second $\#$ on the left border.
The least common ancestor edge is the ingoing edge to this sharp symbol,
colored in red in Fig.~\ref{pic:stripes-grass-left}.
With top-down determinism, this edge will always have the same edge label,
let it be $\stateWB$.
Since neither the dark vertex depends on the light one in the partial order,
nor vice versa, there must be a rule hardcoding the color of the first row.

For two colors, this can be either
  $(\stateWB \sharparrow \stateBW\stateW)$
  or
  $(\stateWB \sharparrow \stateBW\stateB)$.
Let
$L_\squaretopblack$ and $L_\squarebotblack$ denote the languages
with one of the rules above, respectively, hardcoding the first row's color.
Obviously, $\Lstripes = L_\squaretopblack \cup L_\squarebotblack$.
The union of these languages requires both rules in the rule set,
yielding a nondeterministic DAG automaton.
%
% - see stripe language in Example~\ref{ex:stripes}
%   that there is some DAG lang we cannot accept: union of stripes
%   without hardcoded first color Fig.~\ref{pic:stripes-grass-left}
%
%
% ====================================================
%
% - remark for other langs: $\linrDAG$ next lines have to guess because root but if input dependant ok
%
% - remark for other langs: $\linrDAG$ every lines can propage to subsequent lines
%
% - remark for other langs: $\linrDAG$ every lines has to guess, only at line end acceptance
%
\end{proof}
%
%
%
%\begin{corollary}[$\DDA \nequiv_\text{poset} \NDA$]\label{theorem:poset}
%On total orders $\DDA$ and $\NDA$ are of same power, on proper partial orders, not.
%\end{corollary}
%\begin{proof}
%According to \ref{theorem:linlDAGdda-subset-linrDAGdda},
%DDA $\neqlinl$ $\NDA$. Thus, there exists a DAG encoding
%inducing a partial order where $\DDA$ and $\NDA$
%do not recognize the same class of picture languages.
%\end{proof}

%\noindent{\bf Proof of Theorem~\ref{theorem:hatPeqP}}
\theoremhatPeqP*
\begin{proof}
Given an $\NDA$
$\hat A=(\hat N,\Sigmasharp,\hat R)$ for recognizing
pictures $P\in\picdags$ encoded as $\cooDAG(\hat P)$,
there exists an equivalent $\NDA$ $A'=(N',\Sigmasharp,R')$
recognizing the same set of nonempty pictures as $\hat A$
but labeling the edges adjacent to boundary vertices
(i.e., those with vertex label \#)
only with a fixed state $\qzfdag\in\hat N'$.
This results in a uniform labeling of the boundary:
vertices labeled by the \#, adjacent edges by $\qzfdag$.

%If $\hat A$ accepts the empty picture $\Lambda$, we simply add the rules $
%(\lambda \rarrow{\#} \qzfdag\qzfdag),
%      (\qzfdag \rarrow{\#} \qzfdag)$, and
%     $(\qzfdag\qzfdag \rarrow{\#} \lambda)
%$ to the set of rules $R'$.
%This holds because DAG automata operate locally: for each vertex $v \in V$, acceptance depends only on its label $\ell(v) \in \Sigma \cup \{\#\}$ and the tuple of states labeling its incoming and outgoing edges.

The rest of the rules in $R'$ can be constructed so that, for a nonempty picture $P \in \Sigma^{m,n}$, $m\ge 1$ and $n \ge 1$:
\begin{itemize}
\item For the top-left corner picture position $\hat P_{1,1}$, a labeling of six ingoing and outgoing edges of $\hat P_{0,1}$ and $\hat P_{1,0}$ done by the DAG automaton $\hat A$ can be guessed and stored in the labeling of the two outgoing edges from the vertex corresponding to the symbol $\hat P_{1,1}$ using a rule in $R'$. Actually, only the four outgoing edges must be stored.

\item For all positions of $\hat P_{i,j}$ adjacent to the left and top border (except the corners), that is for $(i\in \{2,\ldots,m-1\} \wedge j=1) \vee (i=1 \wedge j \in \{2,\ldots,n-1\})$ there must be rules in $R'$ that guess labeling of the outgoing edges of the closest border vertex by $\hat A$, store it in the states assigned to outgoing edges from the vertex $\hat P_{i,j}$ and check the compatibility with the labeling of ingoing edges of vertices corresponding to $\hat P_{i,j}$ and the border symbol. E.g., the rules in $R'$ that can be applied to $\hat P_{2,1}$ must ensure compatibility of the labeling of the ingoing edge from $\hat P_{1,1}$, which also has the information on the labeling of the edge from $\hat P_{1,0}$.

\item For the top-right corner position $\hat P_{1,n}$ rules in $R'$ must ensure compatibility with possible labeling of edges around the top-right border symbol and the vertex corresponding to the boundary position $\hat P_{1,n+1}$. The label on the edge between $\hat P_{1,n}$ and $\hat P_{2,n}$ must encode labeling of outgoing edges from $\hat P_{1,1}$ and $\hat P_{1,n+1}$.

    Similarly, the rules in $R'$ applicable to the vertex corresponding to the symbol $\hat P_{m,1}$ will encode the labeling around $\hat P_{m,1}$ and $\hat P_{m+1,1}$ done by $\hat A$.

\item For all positions of $\hat P_{i,j}$ adjacent to the bottom and right border (except the bottom-right corner), that is for $(i=m \wedge j \in \{1,\ldots,n-1\}) \wedge (i\in \{2,\ldots,m\} \vee j=n)$ there must be rules in $R'$ that guess labeling of the edges going to and from the closest border symbol assigned by the automaton $\hat A$, store it in the labeling of the outgoing edge which does not enter the border symbol and check their compatibility with the labeling of the preceding border symbol stored in the label of the preceding non-border symbol.

\item The rules in $R'$ applicable to the bottom-right corner vertex corresponding to $\hat P_{m,n}$ must ensure that there is a possible labeling by the automaton $\hat A$ also for the border symbols $\hat P_{m,n+1}$, $\hat P_{m+1,n}$ and $\hat P_{m+1,n+1}$.
\end{itemize}

In all rules in $R'$, all states assigned to the edges from or to a border symbol can be the same fixed state $\qzfdag\in\hat N'$.

In the next step, from $A'$, we can easily construct
an $\NDA$ $A=(N,\Sigma,R)$,
recognizing the same set of pictures as
$\hat A$ and $A'$
but encoded without the boundary
as $\cooDAG(P)$.
This boundary invariance, however, does not apply to one-column
and one-row pictures, thus to strings in vertial or horizontal orientation.
The reason is obvious: a string DAG has no orientation in a 2D-plane.
The boundary adds this orientation.
The equivalence of $A$ and $A'$ is thus restricted to pictures of
the dimensions $(m,n)$ where $m,n > 1$.
With boundary (for $\hat A'$), every picture vertex $P(i,j)$,
for $i \in [m]$ and $j \in [n]$, has in- and out-degree two and,
in outmost picture vertices,
it encodes the side of an outmost picture vertex.
In contrast, the degree of a boundary vertex does not reflect
on which of the four sides a vertex is located: top, bottom, left, or right.
However, this can be encoded into the edge labels, since the root rule
$\lambda\sigmaarrow q\outdown q\outright$ applied by NDA $A$ can be made  aware of the sides top and left
and can propagate that information through the states.
With this information, the corner rules $q\sigmaarrow q'$ of $A$ can detect
the bottom and right side and propagate it accordingly.
Accordingly, the encoding $\cooDAG$ is aware of its four sides,
regardless of whether we encode a picture or a boundary picture with it
This completes the proof that the boundary
does not enhance the expressive power for an $\NDA$
when accepting pictures encoded by $\cooDAG$.
\end{proof}

%\noindent{\bf Proof of Theorem~\ref{theorem:notaeqda}}
\theoremnotaeqda*
\begin{proof}

% First direction: NDA A simulates 2OTA.
We establish the equivalence with respect to picture languages
by mutual simulation of the nondeterministic versions
of online tessellation and DAG automaton.
As the first direction, we present the technically more involved simulation
of an $\NDA$ by a $\twoOTA$.

Let $A = (N, \Sigma \cup \{\#\}, R)$ be an NDA accepting a picture language $L$ using the picture-to-DAG encoding $\cooDAG$ (encoding an input picture $P$ as $\cooDAG(\hat P)$). According to Lemma \ref{theorem:hatPeqP}, we can assume that edges of each boundary vertex (labeled with $\#$)
are labeled with a fixed state, but here we assume that two different states are used. The first one, $\qzdag$, for labeling the edges incident with the top and left border of $\hat P$ and a second one, $\qfdag$, for labeling the edges incident with the right and bottom border. These states are not necessarily distinct states.
They are distinct if and only if
the language excludes the empty picture.
The left and top border vertices are labeled by~$\qzdag$,
including the corner vertices,
the remaining edges of the right and bottom boundary vertices by~$\qfdag$.
Formally, this yields for a rule of the form
$\sharprule$ the six rules:
  $\lambda\sharparrow \qzdag\qzdag$,
  $\qzdag\sharparrow \qzdag\qzdag$,
  $\qzdag\sharparrow \qzdag$,
  $\qzdag\qfdag\sharparrow \qfdag$,
  $\qfdag\qfdag\sharparrow \qfdag$ and
  $\qfdag\qfdag\sharparrow\lambda$.
No other rules are allowed for the boundary. If $\qzdag \ne \qfdag$, for an empty picture $\Lambda$, there is no run, as the rule $\qfdag\qfdag\sharparrow\lambda$ cannot be used for the bottom-right border symbol $\#$.

Now, we show how to construct a 2OTA automaton $M=(\Sigma, Q,q_0,F,\delta)$ accepting the same picture language $L$ by simulating the automaton $A$.

%We construct $M$ so that it never enters its initial state $q_0$ on a position different from the top and left border of an extended picture $\hat P$, for all input pictures $P \in \Sigma^{*,*}$.

Given an $\NDA$ $A=(\dagalph, \Sigmasharp, R)$
with $\symbolqzdag,\symbolqfdag \in \dagalph$,
we construct the $\twoOTA$,
simulating it for pictures, as a five tuple
$M=(\Sigma, Q, q_0, F, \delta)$
as follows:
\begin{enumerate}

  \item A $\twoOTA$ runs on a boundary picture, but it
does, however, never reference the symbols of the boundary,
only their states.
Therefore, it uses solely the \emph{alphabet} $\Sigma$,
excluding $\#$.
\item
The \emph{set of states} $Q = (\dagalphlambda) \cdot (\dagalphlambda)$
comprises the concatenation of two edges labels.
A state $nn' \in Q$ keeps track of two states of $A$ that label the outgoing edges of a node.
% Since the $\twoOTA$ propagates those two states for the two outgoing edges redundantly in both the directions instead of just the specified one, thus, n for downwards and n′ to the right. The edge label in the tuple not complying with the direction is a ‘don’t care’ within the simulation.
Since the $\twoOTA$ propagates those two states for the two outgoing edges
redundantly in both directions instead of just the specified one, thus,
$\qdown$ for downwards \emph{and} $\qright$ to the right,
the edge label in the tuple not complying with the direction
is a `don't care' within the simulation.
See the subsequent construction of the transition relation
in Point~\ref{item:ota:delta} below. For clarity, let us denote the state $\qfdag\qfdag$ as $q_f$.
\item
The \emph{starting state} is $q_0 = ~ \qzdag\qzdag$.

\item\label{item:ota:finalstates}
  The set of \emph{accepting states} is $F=\{q_f\}$. The initial state $q_0$ is final if and only if $q_0=q_f$ and $A$ accepts the DAG $\cooDAG(\hat \Lambda)$.
%  $F = \begin{cases}
%          \{\q_f\}, & \mbox{if } \Lambda \not\in L \\
%          \{\q_0, q_f\}, & \mbox{otherwise}.
%        \end{cases}
%        $$
%  The state $q_f$ corresponds to the rules in $R$
%%with $\alpha_1,\alpha_2,\beta_1,\beta_2\in N \cup \{\lambda\}$
%that can be applied to vertices without outgoing edges (bootom-right corner of an input picture).
%$$\OTAsimDAfinal$$
%Also for deterministic rules, the accepting states can easily be identified,
%see below the section about determinism.
%If $A$ accepts the empty picture $\Lambda$, then the initial state $q_0$ of $M$ is in the set of final states $F$ to ensure that $M$ accepts $\Lambda$. Hence, $M$ accepts the empty picture if and only if $A$ accepts the empty DAG.

\item\label{item:ota:delta}
The \emph{transition relation} $\delta$ computes possible states at a position in the picture based on the states of its adjacent upper and left positions and the symbol at the position. The value of the transition function corresponds to the rule of $\NDA$ $A$ matching the vertex's symbol and its ingoing edges:
$$\OTAsimDAforalltext.$$

% Note that $\delta$ does not need to encode the rules for the borders,
% $\sharprule$,
% since these correspond to identifying the start and accepting states.
%\item
%calculating states corresponds to coloring the outgoing edges
\end{enumerate}

We now prove the correctness of this construction by showing
%using structural induction by showing
%Thus we show that for the size of the picture $(m,n)$
\begin{equation}\label{eq:MA}
M \text{ accepts } P \iff A \text{ accepts } \cooDAG(\hat P)
\quad \text{for all } P \in \Sigmastarstar.
\end{equation}

We will reference the edges via the indices $i,j$ over the dimensions,
just as we do for vertices. We regard the ingoing edges as not part of
a vertex but only its outgoing
since that is the way a $\twoOTA$ works.
Therefore, let $\eijd$ and $\eijr$ denote the outgoing edges,
pointing down and to the right, respectively, such that
$\OUT(\Pij) = \eijd \eijr$.
Note that with this notation $\IN(\Pij) = \eijl\eiju$.

For the dimensions $(m,n) = (0,0)$, the empty picture $P = \Lambda$
encoded as $\cooDAG(\hat \Lambda)$ is a DAG. While $\Lambda$ itself is empty, its boundary picture $\hat\Lambda$
is a $2 \times 2$ grid filled entirely with~\# (recall Def.~\ref{def:pic}). $M$ assigns the initial state to the border positions of $\hat \Lambda$ except the bottom-right corner. The empty input picture is accepted by $M$ if and only if the state $q_0$ assigned to $\hat P(0,0)$ is a final state. Based on the above construction, $q_0$ is final (and equals $q_f$) for $M$ if and only if $A$ accepts $\cooDAG(\hat \Lambda)$.

The relevant rule $\lambda\sharparrow\!\qzdag\qzdag$, translates to
$\delta( n\lambda, \lambda n'\!\!, \qzfdag ) =
\delta( n, n'\!\!, \qzfdag ) =~\qzdag\qzdag$
  which does not appear in $\delta$,
  since, contrary to the $\NDA$,
  the $\twoOTA$ does not operate on the boundary symbols.
  However, this DAG rule yields an accepting state
  if and only if $\qzdag~=~\qfdag$.
  In that case, we can conclude that the starting state
  $\qzdag\qzdag~=~\qfdag\qfdag~= \qzfdag\qzfdag$.
  Then,
$\delta( n, n'\!\!, \# ) = \qzfdag\qzfdag$
  yields the accepting state $\qzfdag\qzfdag \in F$ which is the starting state,
  otherwise not.
  A $\twoOTA$ accepts a picture of dimensions $(m,n)$
  if $q_{m,n}\in F$, which for $\Lambda$ results in $q_{0,0}\in F$.
  The state $q_{0,0} = q_0 = \qzfdag\qzfdag \in F$ iff $A$ accepts $\Lambda$.
%$M$ accepts the empty picture $\Lambda$ if and only if
%$\qzdag = \qfdag$, denoted by $\#$.
  Consequently, $A$ accepts $\Lambda$ if and only $M$ does, too.

For a nonempty input picture $P$, of dimensions $(m,n)$ with $m,n>0$,
the rules $R$ of $\NDA$ $A$ applied to the picture itself, not the boundary,
correspond to the transition relation $\delta$
of the $\twoOTA$.
Let $\rhoDA$ denote a fixed run of the $\NDA$ $A$ on $\cooDAG(\hat P)$
and $\rhoOTA$ the corresponding run
of the $\twoOTA$ $M$ on $P$.
While $\rhoDA$ assigns two states $q_1$ and $q_2$
to the outgoing edges of any vertex $\Pij$ of a picture
according to the rules $R$,
$\delta$ assigns those states
as one state $q_1q_2$ to the vertex $\Pij$.
While $\delta$ assigns only states to the picture itself,
$\rhoOTA$ additionally assigns the edge labels $q_1q_2 \in Q$
as one state to
the border $P(0,j)$ and $P(i,0)$, identically to $\rhoDA$
which labels those outgoing edges with the states $q_1$ and $q_2$.

For each accepting run $\rhoDA$ of the automaton $A$ on $\cooDAG(\hat P)$, we can construct an accepting run $\rhoOTA$ on $\hat P$. The run $\rhoDA$ assigns two states $q_1$ and $q_2$
to the outgoing edges of each vertex $V_{i,j}$ of $\cooDAG(\hat P)$, where $i \in [m]_0$ and $j \in [n]_0$, according to the rules from $R$. Then, $\delta$ assigns those states as one state $q_1 q_2$ to the position $P(i,j)$:

\begin{enumerate}
  \item The string $\qzdag\qzdag$ labels the outgoing edges of
      all vertices $V_{i,0}$ and $V_{0,j}$, where
      $i\in[m]_0$ and $j\in[n]_0$. The automaton $M$ assigns the same tuple $\qzdag\qzdag=q_0$ to all positions $P(i,0)$ and $P(0,j)$.
      The corners $P(m+1,0)$, $P(0,n+1)$ are \emph{not} included.
      They are irrelevant since their assignments are identical for every picture,
      and the picture itself does not depend on them.
      Unlike the $\NDA$,
      the $\twoOTA$ does not even
      reference the states assigned to these positions.
      \[
      \begin{array}{rcl}
      \forallzeroij
      \rhoDA(e_{i,0}\outdown e_{i,0}\outright)
      & = &
      \rhoDA(e_{0,j}\outdown e_{0,j}\outright)
      = \
      \qzdag\qzdag \
      =
      q_{0}\\
      & = &
      q_{i,0} = q_{0,j}
      =
      \rhoOTA(P(i,0)) = \rhoOTA(P(0,j))
      \end{array}
      \]
  \item\label{item:rho:P}
      Now, using induction on $i$ and $j$, we can show that if $\IN(V_{i,j}) = \eijl\eiju$ and $\OUT(V_{i,j}) = \eijd \eijr$, then
      for $\alpha_1 = \rho_A(\eijl), \alpha_2 = \rho_A(\eiju), \beta_1 = \rho_A(\eijd), \beta_2 = \rho_A(\eijr)$ and $\sigma=P_{i,j}$, we have $(\alphaone\alphatwo \sigmaarrowX \betaone\betatwo) \in R$, $\betaonetwo\in
      \delta\left( \statealphaone, \statealphatwo, \sigma \right) = \delta(q_{i,j-1},q_{i-1,j},P_{i,j})$ and we can assign the state $q_{i,j}= \betaone\betatwo$ to position $P_{i,j}$.
  \item As the state assign to both outgoing edges of the vertex $V_{m,n}$ must be $\qfdag$, the state $q_{m,n}$ must be $\qfdag\qfdag\ \in F$. Thus, we get an accepting run of $M$ on $\hat P$.
\end{enumerate}

Similarly, we can show that for each accepting run $\rhoOTA$ of the automaton $M$ on $\hat P$ we can construct an accepting run $\rhoDA$ of the automaton $A$ on $\cooDAG(\hat P)$. This concludes the proof that (\ref{eq:MA}) is true and the picture language $L$ accepted by $\NDA$ $A$ using encoding $\cooDAG$ equals the language accepted by 2OTA $M$.\medskip

%DAG Automaton simulates OTA.
For the second direction,
let
$M = (\Sigma, Q, \qzdag, F, \delta)$ be a $\twoOTA$, where $\qzdag \in Q$ is the initial state of $M$. Without loss of generality, we can assume that $M$ assigns the initial state $\qzdag$ to the top and left border symbols of an extended picture, but not to any position with a symbol from $\Sigma$.
We can construct a $\NDA$ $A$ simulating $M$ as a triple
$A = (N, \Sigmasharp, R)$, in the following way:
\begin{enumerate}
  \item The \emph{set of states} $N = Q \cup \{\qfdag\}$, also called \emph{edge labels}, for a new state $\qfdag\ \not\in Q$,
      uses directly the $\twoOTA$'s states
      and requires additionally an edge label $\qfdag$ for acceptance
      due to the lack of a set of designated final states.

  \item Contrary to an $\twoOTA$ whose transition relation $\delta$
      defines its run $\rho_M$ only on the positions in an input picture $P$ itself,
      the rules $R$ determine a run $\rho_A$ on every edge of the DAG $\cooDAG(\hat P)$.
      Therefore, it uses the \emph{alphabet} $\Sigmasharp$,
      including the boundary symbol $\#$.

  \item The \emph{set of rules} $R$ reflects the starting state $\qzdag$
      and accepting states $F$
      of the $\twoOTA$
      as well as the transition relation $\delta$.
      \begin{itemize}
        \item %starting states yield the rules TODO
          As in the previous simulation, we achieve an equivalence concerning the runs
          on the top and left boundaries by adding the rules
          $\lambda\sharparrow \qzdag\qzdag$,
          $\qzdag\sharparrow \qzdag\qzdag$, and
          $\qzdag\sharparrow \qzdag$ to the set of rules $R$ of $A$.

        \item The rules for the position inside the picture itself, excluding the boundary, are constructed from $\delta$
            by propagating $\twoOTA$'s computed state redundantly.
            Instead of assigning it once to the current position, as $M$ does,
            $A$ labels the outgoing edges
            pointing downwards and to the right equivalently.
            Note that the order of the states that it depends on is switched:
            $$\text{if }
              q \in \deltaqarrows
              \text{, then we add the instruction }
              (q\inright q\indown \sigmaarrow qq) \text{ to } R.
            $$

        \item The $\twoOTA$ $M$ does not run on the boundaries below and right.
            Consequently, $A$ may assign arbitrary states to the edges pointing
            from a picture vertex to such a boundary vertex.
            By the following rules, we allow $A$ to label those edges
            just as edges not pointing to the boundary.
            These rules continue on the right and bottom boundaries
            to use the starting state $\qzdag$ of $M$
            for labeling edges pointing from boundary vertices to boundary vertices.
% Since the first edge of both boundaries is labeled by the corner rules
%   $\qzdag\sharparrow \qzdag$ with the starting state $\qzdag$,
%   allowing to exlude the empty picture,
% also the starting state needs to be considered.
            $$
              \forall q\inright,q\indown \in Q :
              ( \qzdag   q\indown \sharparrow \qzdag),
              (q\inright  \qzdag  \sharparrow \qzdag) \in R
            $$

        \item For agreeing on acceptance,
            $A$ guesses the bottom-right corner position $P(m,n)$ nondeterministically by
            $$\text{if }
              q \in F \land
              q \in \deltaqarrows
              \text{ then we add the instruction }
              (q\inright q\indown \sigmaarrow \qfdag\qfdag) \in R.
            $$
            Also nondeterministically,
            $A$ then guesses both vertices adjacent to the leaf corner
            in order to emulate the acceptance of $M$ by $F$ by adding the following rules to $R$:
            $$
              ( \qzdag \qfdag \sharparrow \qfdag),
              ( \qfdag \qzdag \sharparrow \qfdag) \text{ and }
              (\qfdag\qfdag\sigmaarrow\lambda).
            $$
            Additionally, if the 2OTA $M$ accepts the empty picture $\Lambda$, the starting state $\qzdag$ is final for $M$, and we add the instruction $(\qzdag\qzdag\sigmaarrow\lambda)$ to $R$.
      \end{itemize}
\end{enumerate}

Evidently, the 2OTA $M$ accepts the empty picture if and only if its initial state $\qzdag$ is final and the constructed $\NDA$ $A$ accepts the DAG $\cooDAG(\hat \Lambda)$.

For a nonempty input picture $P$, it is easy to verify that
equivalent runs for $M$ on $\hat P$ and $A$ on $\cooDAG(\hat P)$ can be constructed with the rule set $R$
in a similar way as for the first direction,
agreeing furthermore again on the acceptance of their inputs.

We have shown mutual simulation of the two models.
Thus, $\twoOTA \eqcoo \NDA$.
\end{proof}

%\noindent{\bf Proof of Theorem~\ref{theorem:escapeotasgraffito}}
\theoremescapeotasgraffito*
\begin{proof}
The proof is by construction.
Consider the picture language $L\in\{\mathtt{a},\mathtt{b}\}^{*,*}$
with a balanced number of positions labeled by $\mathtt{a}$ and $\mathtt{b}$.
Formally:
\[
\Lbalance = \Bigl\{\: P \in \Sigma^{m,n} \Bigm| \sum_{i=1}^{m} \sum_{j=1}^{n} |P_{i, j}|_\mathtt{a} = \sum_{i=1}^{m} \sum_{j=1}^{n}|P_{i, j}|_\mathtt{b} \text{ and } m,n \in \mathbb{N}\:\Bigl\}.
\]
Neither a $\twoOTA$ nor a sgraffito automaton can recognize this language,
because both automaton models lack the ability to count.
But, a DAG automaton $A =( \,\{p\},(\Sigma,r), R)$, where $p$ is a state and
$R =
(\lambda \rarrow a p),
(p \rarrow b \lambda )
$, accepts a picture $P$
if and only if $P\in L$.
Each edge labeled with the state $p$ pairs exactly one $\mathtt{a}$ with one $\mathtt{b}
$.
The input-driven encoding yields
unconnected picture DAGs where every connected component
consists of two vertices, labeled $a$ and $b$ respectively, joined by an edge
labeled~$p$.
\end{proof}

\end{document}